\documentclass[sort&compress]{elsarticle}

\usepackage{amsmath,amssymb,amsthm,color,paralist,graphics,epsfig,graphicx,epstopdf}

\usepackage{hyperref}

\newcounter{count}

\newtheorem{theorem}{Theorem}[section]
\newtheorem{assumption}[theorem]{Assumption}
\newtheorem{corollary}[theorem]{Corollary}
\newtheorem{definition}[theorem]{Definition}

\newtheorem{remark}[theorem]{Remark}

\begin{document}

\begin{frontmatter}

\title{Open Quasispecies Models: Stability, Optimization, and Distributed Extension}

\author[address_1]{Ivan~Yegorov\corref{mycorrespondingauthor}\fnref{myfootnote}}
\cortext[mycorrespondingauthor]{Corresponding author}
\ead{ivanyegorov@gmail.com}
\fntext[myfootnote]{Also known as Ivan Egorov}

\author[address_1]{Artem~S.~Novozhilov}
\ead{artem.novozhilov@ndsu.edu}

\author[address_2,address_3]{Alexander~S.~Bratus}
\ead{alexander.bratus@yandex.ru}

\address[address_1]{North Dakota State University, 1210 Albrecht Boulevard, Fargo, ND 58102, USA \medskip}

\address[address_2]{Lomonosov Moscow State University, Leninskie Gory, MSU, 2nd educational building, Moscow, 119991, Russia \medskip}

\address[address_3]{Russian University of Transport, Obraztsova 15, Moscow 127994, Russia \medskip}

\begin{abstract}
We suggest a natural approach that leads to a modification of classical quasispecies models and incorporates the possibility of
population extinction in addition to growth. The resulting modified models are called open. Their essential properties, regarding
in particular equilibrium behavior, are investigated both analytically and numerically. The hallmarks of the quasispecies dynamics, viz.
the heterogeneous quasispecies distribution itself and the error threshold phenomenon, can be observed in our models, along with extinction.
In order to demonstrate the flexibility of the introduced framework, we study the inverse problem of fitness allocation under the biologically
motivated criterion of steady-state fitness maximization. Having in mind the complexity of numerical investigation of high-dimensional
quasispecies problems and the fact that the actual number of genotypes or alleles involved in a studied process can be extremely large,
we also build continuous-time distributed open quasispecies models. The obtained results may serve as an initial step to developing
mathematical models that involve directed therapy against various pathogens.
\end{abstract}


\end{frontmatter}


\section{Introduction}

Mathematical modeling has long been a key theoretical tool connecting various pictorial and verbal models of RNA virus evolution,
establishment, and extinction, as well as in vitro and in vivo experiments aimed to understand and potentially fight this rich
group of human pathogens~\cite{domingo2015}. Arguably, the onset of this specific modeling, which was done originally in terms of
self-replicating polynucleotide sequences to study the problem of the origin of life, can be traced back to the founding papers by
M.~Eigen~et~al.~\cite{eigen1971sma,Eigen1988,Eigen1989,Eigen1977,eigenshuster}, as well as by J.\,F.~Crow and
M.~Kimura~\cite{CrowKimura1970,CrowKimura1964,Kimura1965}. They proposed so-called \textit{quasispecies models}, which were later
connected directly to RNA virus evolution~\cite{domingo1978nucleotide}; see \cite{andino2015viral,Domingo2012} for a more recent account.

For such quasispecies models, two key phenomena were observed. The first one was the highly heterogeneous cloud of mutants of the most fit
(or master) sequence at the selection-mutation equilibrium. This cloud was called the quasispecies, hence the name of the model. The second
phenomenon was the so-called \textit{error threshold}. It can be described as the critical mutation rate (or probability, depending on
the settings of a particular model) that divides the selective phase of the virus evolution, i.\,e., the dominance of the master sequence in
the population, from the random phase, where the distribution of various genotypes becomes essentially uniform.

A great deal of mathematical investigation was devoted to study exact properties of the quasispecies and error threshold; see, e.\,g.,
the review in \cite{BratusNovozhilovSemenov2017}. In a nutshell, the exact details of the structure of the mutation-selection equilibrium and
the precise position (if it exists at all) of the error threshold depend in a subtle way on the implemented fitness landscape. They are
in general far from the oversimplified and widely referred formula that the error threshold is typically given by the selective advantage of
the master sequence over the sequence length; see, e.\,g.,~\cite{jainkrug2007,SemenovNovozhilov2015,wiehe1997model} for specific examples.

From a more practical point of view, the error threshold concept led to the idea of lethal mutagenesis, i.\,e., the process of virus extinction
induced by elevated mutation rates~\cite{eigen2002error}, and eventually to first mutagenetic experiments~\cite{cameron2001mechanism}.
The undeniable influence of the error threshold phenomena on the idea of lethal mutagenesis notwithstanding, there existed an internal
contradiction between these two concepts from the modeling point of view. The fact is that the quasispecies models were formulated as systems of
ordinary differential equations (ODEs) describing the distribution of the relative frequencies of the master sequence and its various mutants,
but not their population sizes (or densities with respect to a certain external measure). It was therefore simply meaningless to discuss
possible extinction within the framework of these models~\cite{bull2007theory,wilke2005quasispecies}. In~\cite{bull2007theory}, it was reasonably
noticed that the process of the loss of the master sequence in the population (the error threshold) should in general be distinguished from
the actual extinction of a whole virus population, and the original quasispecies models would have to be modified to incorporate
the possibility of population extinction. The paper~\cite{bull2007theory} was the first to suggest such a modification. Under a number of
simplifying assumptions (e.\,g., the multiplicativity of the fitness landscape, no back or compensatory mutations, etc.), the following simple
mathematical condition for population extinction was obtained: the product of the absolute fitness of the master sequence and the average
biological fitness should be less than one; see also~\cite{tejero2015theories}.

Other examples of quasispecies models including the possibilities of lethal mutations and population extinction were given in
\cite{chen2009lethal,martin2010lethal,tejero2010effect}. They were reviewed and put in a general context in \cite{tejero2015theories}.
Among other things, the analysis (analytical and/or numerical) of these models showed, quite naturally, that taking into account one or
another more realistic assumption would lead to corrections of the original extinction criterion obtained in \cite{bull2007theory}.
Moreover, analytical results are usually possible only under some significant simplifications regarding the fitness and mutational
landscapes.

The main goal of the present text is to develop a relatively simple but yet general and flexible modeling framework
incorporating growth and mortality characteristics, as well as to investigate its basic properties. Due to the possibility of population
extinction, our models can be called open, in contrast to the classical quasispecies models (formulated in terms of relative frequencies) which
can be called closed. Our approach differs from those presented in the aforementioned
works~\cite{bull2007theory,chen2009lethal,martin2010lethal,tejero2010effect,tejero2015theories}, since we do not start with any radical
simplifying assumptions, e.\,g., we do not fix the fitness or mutational landscapes. Our requirement is to stay connected with the classical
quasispecies models as close as reasonably possible. Among other things, such a requirement potentially yields the opportunity to support our
study with a number of existing mathematical techniques for the classical quasispecies models; see, e.\,g.,
\cite{Burger2000,cerf2016quasispecies,semenov2017generalized,semenov2016eigen,SemenovNovozhilov2015}.

In order to demonstrate the flexibility of our framework in case of ODE-based dynamics, we also study the inverse problem of fitness
allocation under the biologically motivated criterion of steady-state fitness maximization.

Furthermore, having in mind the complexity of numerical investigation of high-dimensional quasispecies problems and the fact that
the actual number of genotypes or alleles involved in a studied process can be extremely large, we propose continuous-time distributed open
quasispecies models with growth and mortality. Note that the first particular distributed version of a classical quasispecies model was
in fact introduced in \cite{CrowKimura1964,Kimura1965}. It was formulated as an integro-differential equation with respect to the sought-after
time-varying probability density describing the relative frequencies of considered genotypes. A rigorous mathematical investigation of
a class of such distributed models was provided in \cite{Burger1986,Burger1988a,Burger1988b,Burger1996,Burger2000}. The dynamics in our
distributed formulation is not restricted to the probability density constraint and takes growth and mortality characteristics into account.

Despite the mostly mathematical content of the current work, it has to be emphasized that the general methodology leading to our
approach can be applied to a variety of different biological systems, including not only viruses but also bacteria, cancer cells, etc.
Instead of descriptions of particular genotypes, more aggregated quantities may in principle be treated in quasispecies models.
In particular, we keep in mind future applications of the proposed framework to developing mathematical models that involve directed
therapy against various pathogens, such as infections or cancer~\cite{KomarovaWodarz2014,WodarzKomarova2014,SchattlerLedzewicz2015},
which will be considered elsewhere.

The rest of the paper is organized as follows. In Section~\ref{Sec_2}, we construct a class of ODE-based open quasispecies models and establish
its essential properties. Section~\ref{Sec_3} investigates the corresponding steady-state fitness maximization problem. Sections~\ref{Sec_2} and
\ref{Sec_3} also contain some related numerical simulation results. Section~\ref{Sec_4} is devoted to a class of continuous-time distributed
open quasispecies models. A summarizing discussion of possible future developments is given in Section~\ref{Sec_5}.

\section{ODE-based open quasispecies models}
\label{Sec_2}

\subsection{Classical Eigen and Crow--Kimura quasispecies models}

To make the text self-containing, we start with formulating the classical ODE-based quasispecies models.

The classical \textit{Eigen quasispecies model} \cite{Eigen1988,Eigen1989} is formulated via the system of ordinary differential
equations
\begin{equation}
\dot{p}(t) \:\, = \:\, K \, M \, p(t) \: - \: f[p(t)] \, p(t), \quad t \geqslant 0,  \label{Eq_1}
\end{equation}
where the following notations are used:
\begin{itemize}
\item  $ t \geqslant 0 $ is a time variable;
\item  $ p(t) \: = \: (p_1(t), \, p_2(t), \, \ldots, \, p_n(t))^{\top} \: \in \: [0, +\infty)^{n \times 1} \: $ is the vector of the relative
(normalized) time-varying frequencies of $ n $ considered genotypes (or allelic effects) labeled by $ 1, 2, \ldots, n $,
$ \: \sum_{i = 1}^n p_i(t) \, = \, 1 \: $ for all $ t \geqslant 0 $;
\item  $ K \, = \, \{ k_{ij} \}_{i,j \, = \, 1}^n \, \in \, [0, 1]^{n \times n} \: $ is the mutation matrix, where $ k_{ij} $ is
the probability of the reproduction event for which an individual of genotype~$ j $ produces an individual of genotype~$ i $, so that
\begin{equation}
\sum_{i = 1}^n k_{ij} \: = \: 1, \quad j = \overline{1, n};  \label{Eq_2}
\end{equation}
\item  $ M \: = \: \mathrm{diag} \, [m_1, m_2, \ldots, m_n] \: \in \: \mathbb{R}^{n \times n} \: $ is the diagonal matrix whose
diagonal elements form the fitness landscape describing the relationship between the genotypes and reproductive success;
\item  $ f[p(t)] \: = \: \sum_{i = 1}^n m_i \, p_i(t) \: $ is the mean population fitness.
\end{itemize}

The \textit{Crow--Kimura quasispecies model}~\cite{CrowKimura1970} relies on the assumptions that the birth events are error-free
and the mutations occur during the life time of the studied structures, i.\,e., the birth events and mutations are separated
on the time scale. The corresponding dynamical system is written as
\begin{equation}
\dot{p}(t) \:\, = \:\, \left( M + \hat{M} \right) \, p(t) \: - \: f[p(t)] \, p(t), \quad t \geqslant 0,  \label{Eq_3}
\end{equation}
where $ \: \hat{M} \, = \, \{ \mu_{ij} \}_{i,j \, = \, 1}^n \, \in \, [0, +\infty)^{n \times n} $, $ \: \mu_{ij} $ denotes
the rate of mutation of genotype~$ j $ into genotype~$ i $,
\begin{equation}
\mu_{jj} \:\, = \:\, -\sum_{i \: \in \: \{ 1, 2, \ldots, n \} \, \setminus \, \{ j \}} \mu_{ij},  \label{Eq_4}
\end{equation}
and other notations are as in (\ref{Eq_1}). The permutation invariant Crow--Kimura
model~\cite{BratusNovozhilovSemenov2014,SemenovNovozhilov2015} considers $ N + 1 $~classes of genotypes and is obtained from
(\ref{Eq_3}) by replacing $ n $ with $ N + 1 $ and introducing the tridiagonal matrix
\begin{equation}
\begin{aligned}
& \hat{M} \, = \, \mu \, Q, \quad \mu \, > \, 0, \quad Q \, \in \, \mathbb{R}^{(N + 1) \, \times \, (N + 1)}, \\
& Q \:\: = \:\: \begin{pmatrix}
-N & 1 & 0 & 0 & \ldots & \ldots & 0 & 0 \\
N & -N & 2 & 0 & \ldots & \ldots & 0 & 0 \\
0 & N - 1 & -N & 3 & \ldots & \ldots & 0 & 0 \\
0 & 0 & N - 2 & - N & \ldots & \ldots & 0 & 0 \\
\ldots & \ldots & \ldots & \ldots & \ldots & \ldots & \ldots & \ldots \\
0 & 0 & \ldots & \ldots & \ldots & -N & N - 1 & 0 \\
0 & 0 & \ldots & \ldots & \ldots & 2 & -N & N \\
0 & 0 & \ldots & \ldots & \ldots & 0 & 1 & -N
\end{pmatrix}.
\end{aligned}  \label{Eq_5}
\end{equation}

By using the properties~(\ref{Eq_2}),~(\ref{Eq_4}) as well as the definitions of $ M $ and $ f $, one can easily verify that,
for a solution of the system~(\ref{Eq_1}) or (\ref{Eq_3}), the condition $ \: \sum_{i = 1}^n p_i(0) \, = \, 1 \: $ implies
$ \: \sum_{i = 1}^n p_i(t) \, = \, 1 \: $ for all $ t > 0 $.

The models~(\ref{Eq_1}) and (\ref{Eq_3}) are stated for closed (isolated) systems of populations, such that their interactions with
the environment and surrounding populations of other kinds are negligible with respect to the studied processes. The death rates are
either absent or implicitly incorporated in the fitness landscape. Since the state variables are interpreted as relative frequencies
and the state trajectories do not leave the standard simplex, growth and mortality are essentially not considered.

\subsection{Open quasispecies models}
\label{Subsect_2_2}

From the perspective of mathematical modeling of therapy processes against infections or
cancer~\cite{WodarzKomarova2014,KomarovaWodarz2014,SchattlerLedzewicz2015}, it is relevant to investigate how various quasispecies
related to certain pathogens or diseased cells reproduce, mutate, die, and resist to targeted attacks from therapeutic agents.
For that purpose, it is reasonable to build open quasispecies models whose state variables are interpreted not as relative
frequencies, but as densities with respect to a certain external measure, so that their sum does not have to be invariant, and
one can take growth and mortality properties into account. Before developing complicated models that involve the dynamics of
a therapeutic agent and its influence on the dynamics of studied quasispecies, it is useful first to treat the simpler case
when the death rates of the quasispecies are constant and the growth terms include only constant coefficients and a saturation
factor depending on the total population size. This corresponds to the situation when the concentration of a therapeutic agent
is assumed to be nearly constant within the observed time interval and can therefore be excluded from the state variables.
Our work focuses namely on such kinds of open quasispecies models with growth and mortality (see also the discussion in the introduction).

Let us write both systems~(\ref{Eq_1}) and (\ref{Eq_3}) in the common form
\begin{equation}
\dot{p}(t) \:\, = \:\, G \, p(t) \: - \: f[p(t)] \, p(t), \quad t \geqslant 0,  \label{Eq_6}
\end{equation}
with $ \: G \, = \, \{ g_{ij} \}_{i,j \, = \, 1}^n \, \in \, \mathbb{R}^{n \times n} \: $ satisfying
\begin{equation}
\sum_{i = 1}^n g_{ij} \: = \: m_j, \quad j = \overline{1, n},  \label{Eq_7}
\end{equation}
\begin{equation}
g_{ij} \geqslant 0 \:\:\: \mathrm{for} \:\:\: i \neq j, \quad i = \overline{1, n}, \quad j = \overline{1, n}.  \label{Eq_8}
\end{equation}
One has $ \, G = K M \, $ for (\ref{Eq_1}) and $ \, G = M + \hat{M} \, $ for (\ref{Eq_3}).

Note that a solution of (\ref{Eq_6}) with $ \: \sum_{i = 1}^n p_i(0) \, = \, 1 \: $ can be represented as
$$
p(t) \:\, = \:\, \frac{1}{\sum_{i = 1}^n r_i(t)} \: r(t) \quad \forall t \geqslant 0,
$$
where
$$
\begin{aligned}
& r(t) \:\, = \:\, (r_1(t), \, r_2(t), \, \ldots, \, r_n(t))^{\top} \:\, = \:\,
c \: \exp \left\{ \int\limits_0^t f[p(\tau)] \, \mathrm{d} \tau \right\} \, p(t) \\
& \forall t \geqslant 0
\end{aligned}
$$
is the unique solution of
\begin{equation}
\dot{r}(t) \: = \: G \, r(t), \quad t \geqslant 0,  \label{Eq_9}
\end{equation}
$$
r(0) \: = \: c \, p(0),
$$
with an arbitrary constant $ \: c \, \in \, \mathbb{R} \setminus \{ 0 \}, \: $ i.\,e., the components of $ r $ may be
understood as some size or density variables whose normalization leads to the relative frequencies constituting $ p $
(a similar fact was originally noticed in \cite{thompson1974eigen}).

We hence proceed from the system~(\ref{Eq_9}) and incorporate growth and mortality characteristics in it. The following
notations are adopted:
\begin{itemize}
\item  $ t \geqslant 0 $ is a time variable;
\item  $ u(t) \: = \: (u_1(t), \, u_2(t), \, \ldots, \, u_n(t))^{\top} \: \in \: [0, +\infty)^{n \times 1} \: $ is
the vector of the densities of $ n $ considered genotypes (or allelic effects) with respect to a certain external measure;
\item  $ m \: = \: (m_1, m_2, \ldots, m_n)^{\top} \: \in \: \mathbb{R}^{n \times 1} \: $ is the corresponding
fitness landscape, $ \: M \: = \: \mathrm{diag} \,\, m \: \in \: \mathbb{R}^{n \times n} $;
\item  $ d \: = \: (d_1, d_2, \ldots, d_n)^{\top} \: \in \: [0, +\infty)^{n \times 1} \: $ is the vector of
the corresponding death rates, $ \: D \: = \: \mathrm{diag} \,\, d \: \in \: [0, +\infty)^{n \times n} $;
\item  the growth and mutation properties are described by a matrix~$ G \in \mathbb{R}^{n \times n} $ satisfying
(\ref{Eq_7}),~(\ref{Eq_8}) as well as by a growth saturation term $ \, \varphi \left( \sum_{i = 1}^n u_i(t) \right) \, $
depending on the total population size $ \, \sum_{i = 1}^n u_i(t), \, $ where
$ \: \varphi \colon \, \mathbb{R} \to [0, +\infty) \: $ is an appropriate function.
\end{itemize}
Thus, we arrive at the autonomous system
\begin{equation}
\dot{u}(t) \:\: = \:\: \varphi \left( \sum_{i = 1}^n u_i(t) \right) \, G \, u(t) \,\, - \,\, D \, u(t), \quad t \geqslant 0.
\label{Eq_10}
\end{equation}

A similar way of introducing a growth saturation term was implemented in \cite{BratusLukasheva2009,Pavlovich2012} as applied to
open replicator systems.

The origin~$ u = 0 $ is a trivial steady state of (\ref{Eq_10}). Set the initial condition for (\ref{Eq_10}) as
\begin{equation}
u(0) \: = \: u^0 \: = \: \left( u^0_1, u^0_2, \ldots, u^0_n \right)^{\top} \: \in \: [0, +\infty)^{n \times 1}.
\label{Eq_11}
\end{equation}

By $ \: F \: = \: (F_1, F_2, \ldots, F_n)^{\top} \, \colon \: \mathbb{R}^{n \times 1} \to \mathbb{R}^{n \times 1}, \: $
denote the function in the right-hand side of (\ref{Eq_10}):
\begin{equation}
\begin{aligned}
& F_i(u) \:\: = \:\: \varphi \left( \sum_{j = 1}^n u_j \right) \, \sum_{j = 1}^n g_{ij} u_j \,\, - \,\, d_i u_i \\
& \forall \: u \: = \: (u_1, u_2, \ldots, u_n)^{\top} \: \in \: \mathbb{R}^{n \times 1}, \quad i = \overline{1, n}.
\end{aligned}  \label{Eq_12}
\end{equation}
If $ \varphi $ is differentiable, the related Jacobian matrix takes the form
\begin{equation}
\begin{aligned}
& \mathrm{D} F(u) \: = \: \left\{ \frac{\partial F_i(u)}{\partial u_j} \right\}_{i, j \, = \, 1}^n \, , \\
& \frac{\partial F_i(u)}{\partial u_j} \:\: = \:\: \begin{cases}
\varphi' \left( \sum\limits_{\nu = 1}^n u_{\nu} \right) \, \sum\limits_{\nu = 1}^n g_{i \nu} u_{\nu} \:\, + \:\,
\varphi \left( \sum\limits_{\nu = 1}^n u_{\nu} \right) \, g_{ii} \:\, - \:\, d_i, & i = j, \\
\varphi' \left( \sum\limits_{\nu = 1}^n u_{\nu} \right) \, \sum\limits_{\nu = 1}^n g_{i \nu} u_{\nu} \:\, + \:\,
\varphi \left( \sum\limits_{\nu = 1}^n u_{\nu} \right) \, g_{ij}, & i \neq j,
\end{cases} \\
& \forall \: u \: = \: (u_1, u_2, \ldots, u_n)^{\top} \: \in \: \mathbb{R}^{n \times 1}, \quad i = \overline{1, n}, \quad
j = \overline{1, n}.
\end{aligned}  \label{Eq_13}
\end{equation}

\begin{assumption}  \label{Ass_1}
$ m ${\rm ,} $ M ${\rm ,} $ d ${\rm ,} $ D ${\rm ,} and $ G $ are as given above in this subsection
{\rm ((\ref{Eq_7})} and {\rm (\ref{Eq_8})} hold in particular{\rm )}. $ \varphi \colon \, \mathbb{R} \to [0, +\infty) \: $
is a nonnegative continuously differentiable function such that the functions
$ \: [0, +\infty) \, \ni \, s \: \longmapsto \: s \, \varphi(s) \: $ and
$ \: [0, +\infty) \, \ni \, s \: \longmapsto \: s \, \varphi'(s) \: $ are bounded.
\end{assumption}

\begin{theorem}  \label{Thm_2}
Let Assumption~{\rm \ref{Ass_1}} hold. Then the nonnegative orthant $ [0, +\infty)^{n \times 1} $ is positively invariant
with respect to the system~{\rm (\ref{Eq_10}) (}i.\,e.{\rm ,} any state trajectory of {\rm (\ref{Eq_10})} starting in
this orthant stays there further up to the right end of the largest time interval of definition{\rm )}. Moreover{\rm ,}
for any initial state~{\rm (\ref{Eq_11}),} there exists a unique solution of {\rm (\ref{Eq_10})} that is defined for
all~$ t \geqslant 0 $.
\end{theorem}

Theorem~\ref{Thm_2} is proved in Appendix.

\begin{remark}  \rm  \label{Rem_3}
A particular admissible choice of the growth saturation function~$ \varphi $ is
\begin{equation}
\varphi(s) \, = \, e^{-\gamma s} \quad \forall s \in \mathbb{R}, \qquad \gamma \, = \, \mathrm{const} \, > \, 0.
\label{Eq_14}
\end{equation}
The case when $ \: \varphi(s) \, = \, \alpha \, e^{-\gamma s} \: $ for all $ s \in \mathbb{R} $ and $ \alpha, \gamma $
are positive constants is trivially reduced to (\ref{Eq_14}) by substituting $ \: \varphi / \alpha, \, \alpha m, \, \alpha G \: $
with $ \, \varphi, m, G, \, $ respectively.  \qed
\end{remark}

The next assumption is needed in particular to ensure the invertibility of the mortality matrix~$ D $ as well as the boundedness
of the total population size function
\begin{equation}
s(t) \: = \: \sum_{i = 1}^n u_i(t) \quad \forall t \geqslant 0  \label{Eq_15}
\end{equation}
along the solution of (\ref{Eq_10}),~(\ref{Eq_11}).

\begin{assumption}  \label{Ass_4}
$ d_{\min} \: = \: \min_{i \, = \, \overline{1, n}} \, d_i \: > \: 0 $.
\end{assumption}

Assumption~\ref{Ass_4} yields that
\begin{equation}
D^{-1} \:\, = \:\, \mathrm{diag} \, \left[ d_1^{-1}, \, d_2^{-1}, \, \ldots, \, d_n^{-1} \right].  \label{Eq_16}
\end{equation}

\begin{theorem}  \label{Thm_5}
Under Assumptions~{\rm \ref{Ass_1}} and {\rm \ref{Ass_4},} the total population size function~{\rm (\ref{Eq_15})} is bounded
along the solution of {\rm (\ref{Eq_10}),~(\ref{Eq_11})}.
\end{theorem}

Theorem~\ref{Thm_5} is proved in Appendix.

\subsection{Steady-state analysis}
\label{Subsect_2_3}

In this subsection, we characterize a nontrivial steady state of the system~(\ref{Eq_10}) with growth and mortality under
some new conditions in addition to Assumptions~\ref{Ass_1} and \ref{Ass_4}.

\begin{assumption}  \label{Ass_6}
The function~$ \varphi $ is positive and strictly decreasing on $ [0, +\infty) ${\rm ,} and
$ \: l_{\varphi} \: = \: \lim_{s \, \to \, +\infty} \, \varphi(s) \: \geqslant \: 0 $.
\end{assumption}

\begin{remark}  \rm  \label{Rem_7}
Assumption~\ref{Ass_6} implies the existence of the strictly decreasing inverse function
$ \: \varphi^{-1} \, \colon \: (l_{\varphi}, \varphi(0)] \, \to \, [0, +\infty) $. If $ \varphi $ takes
the form~(\ref{Eq_14}), one has
$$
\varphi(0) \, = \, 1, \quad l_{\varphi} \, = \, 0, \quad
\varphi^{-1}(\rho) \,\, = \,\, \frac{1}{\gamma} \, \ln \, \frac{1}{\rho} \,\, \geqslant \,\, 0 \quad
\forall \rho \in (0, 1],
$$
and Assumption~\ref{Ass_6} is fulfilled.  \qed
\end{remark}

Recall also that the diagonal mortality matrix~$ D $ admits the diagonal inverse~(\ref{Eq_16}) in line with
Assumption~\ref{Ass_4}.

A nontrivial steady state $ \: u^* \: = \: (u^*_1, u^*_2, \ldots, u^*_n)^{\top} \: \in \: [0, +\infty)^{n \times 1} \: $
of the system~(\ref{Eq_10}) satisfies
\begin{equation}
D^{-1} G u^* \: = \: \frac{1}{\varphi(s^*)} \, u^*, \quad s^* \: = \: \sum_{i = 1}^n u^*_i \: > \: 0.  \label{Eq_17}
\end{equation}
We therefore arrive at the problem of finding a positive real eigenvalue~$ \lambda^* $ of the matrix~$ D^{-1} G $
with a related eigenvector $ \: u^* \: = \: (u^*_1, u^*_2, \ldots, u^*_n)^{\top} \: \in \: [0, +\infty)^{n \times 1} \: $
such that
\begin{equation}
\varphi(s^*) \: = \: \frac{1}{\lambda^*} \: \in \: (l_{\varphi}, \varphi(0)), \quad
s^* \: = \: \sum_{i = 1}^n u^*_i \: = \: \varphi^{-1} \left( \frac{1}{\lambda^*} \right) \: > \: 0.  \label{Eq_18}
\end{equation}
For such $ \lambda^* $ and $ u^* $, the relations $ \, G u^* = \lambda^* D u^* \, $ and (\ref{Eq_7}) lead to
\begin{equation}
\lambda^* \: = \: \frac{\sum_{i = 1}^n m_i u^*_i}{\sum_{i = 1}^n d_i u^*_i} \, .  \label{Eq_19}
\end{equation}

In the presence of mortality, it is reasonable to update the fitness definition for taking into account that an increase
in the death rates negatively affects the population, as well as for maintaining a convenient representation of
the steady-state fitness in terms of an appropriate eigenvalue similarly to the classical quasispecies models
(for the latter, such a representation is discussed, e.\,g., in \cite[\S IV.3]{Burger2000},
\cite[\S 2]{BratusNovozhilovSemenov2014}, \cite[\S 1, \S 2]{SemenovNovozhilov2015}, and
\cite[\S 1]{BratusNovozhilovSemenov2017}).

\begin{definition}  \label{Def_8}
The population fitness for the modified quasispecies models described by the system~{\rm (\ref{Eq_10})} is given by
\begin{equation}
\begin{aligned}
& \tilde{f}[u] \:\: = \:\: \begin{cases}
0, & \sum_{i = 1}^n u_i \: = \: 0, \\
\dfrac{f[u]}{\sum_{i = 1}^n d_i u_i} \,\, = \,\, \dfrac{\sum_{i = 1}^n m_i u_i}{\sum_{i = 1}^n d_i u_i} \, , &
\sum_{i = 1}^n u_i \: > \: 0,
\end{cases} \\
& \forall \: u \: = \: (u_1, u_2, \ldots, u_n)^{\top} \: \in \: [0, +\infty)^{n \times 1}.
\end{aligned}  \label{Eq_20}
\end{equation}
\end{definition}

Hence, the relation~(\ref{Eq_19}) means that
\begin{equation}
\lambda^* \, = \, \tilde{f}[u^*].  \label{Eq_21}
\end{equation}

One more assumption is required for investigating the eigenvalue problem for the matrix~$ D^{-1} G $. The definitions and
properties of essentially nonnegative matrices, dominant eigenvalues, and irreducible matrices have to be recalled
(see, e.\,g., \cite[\S I.7]{Kato1982}).

\begin{assumption}  \label{Ass_9}
The matrix~$ G $ is irreducible.
\end{assumption}

\begin{remark}  \rm  \label{Rem_10}
The matrix~$ G $ is essentially nonnegative according to (\ref{Eq_8}). Assumptions~\ref{Ass_4}, \ref{Ass_9} and
the relation~(\ref{Eq_16}) imply that the matrix~$ D^{-1} G $ is also essentially nonnegative and irreducible
(see \cite[\S I.7.4]{Kato1982}).  \qed
\end{remark}

\begin{remark}  \rm  \label{Rem_11}
For the matrix~$ \, G = K M \, $ in the Eigen model~(\ref{Eq_1}), Assumption~\ref{Ass_9} holds if, e.\,g.,
$ \, m_i $, $ i = \overline{1, n}, \, $ and all elements of $ K $ are positive (this sufficient condition can be
relaxed in certain cases). For the matrix~$ \, G = M + \hat{M} \, $ in the Crow--Kimura model~(\ref{Eq_3}),
Assumption~\ref{Ass_9} holds if and only if $ \hat{M} $ is irreducible ($ \hat{M} $ is a~priori essentially
nonnegative, since $ M $ is diagonal and $ \, G = M + \hat{M} \, $ fulfills (\ref{Eq_8})). In case of
the permutation invariant Crow--Kimura model specified by (\ref{Eq_5}), $ \, \hat{M} = \mu Q \, $ is indeed
irreducible, because $ Q $ is irreducible. The latter property can be verified with the help of
\cite[\S I.7.4]{Kato1982} and \cite[Fact~2 in \S 9.2]{Hogben2007} (the digraph based criterion of irreducibility for
square matrices with nonnegative elements remains valid for essentially nonnegative matrices).  \qed
\end{remark}

\begin{theorem}  \label{Thm_12}
Under Assumptions~{\rm \ref{Ass_1}, \ref{Ass_4},} and {\rm \ref{Ass_9},} the following properties hold{\rm :}
\begin{itemize}
\item  the matrix~$ D^{-1} G $ has a dominant eigenvalue~$ \lambda^* ${\rm ,} which is real and greater than
the real part of any other eigenvalue of $ D^{-1} G ${\rm ;}
\item  the eigenvalue~$ \lambda^* $ is simple {\rm (}i.\,e.{\rm ,} has algebraic multiplicity one{\rm )} and
admits an eigenvector all of whose components are positive{\rm ;}
\item  there are no other eigenvalues of $ D^{-1} G $ admitting eigenvectors all of whose components are
nonnegative.
\end{itemize}
\end{theorem}

\begin{proof}
It suffices to use Remark~\ref{Rem_10} and \cite[Theorem~I.7.10]{Kato1982}.
\end{proof}

\begin{theorem}  \label{Thm_13}
Let Assumptions~{\rm \ref{Ass_1}, \ref{Ass_4}, \ref{Ass_6},} and {\rm \ref{Ass_9}} hold. Then a nonzero steady state
$ \: u^* \: = \: (u^*_1, u^*_2, \ldots, u^*_n)^{\top} \: \in \: [0, +\infty)^{n \times 1} \: $ of
the system~{\rm (\ref{Eq_10})} exists if and only if the dominant eigenvalue~$ \lambda^* $ of the matrix~$ D^{-1} G $
satisfies
\begin{equation}
\lambda^* \, > \, 0, \quad \frac{1}{\lambda^*} \: \in \: (l_{\varphi}, \varphi(0)).  \label{Eq_22}
\end{equation}
Moreover{\rm ,} if {\rm (\ref{Eq_22})} holds{\rm ,} this nontrivial steady state is uniquely determined as
the eigenvector of $ D^{-1} G $ corresponding to $ \lambda^* $ and such that all of its components are
positive and their sum equals~$ \, \varphi^{-1} (1 / \lambda^*) ${\rm :}
\begin{equation}
s^* \,\, = \,\, \sum_{i = 1}^n u^*_i \,\, = \,\, \varphi^{-1} \left( \frac{1}{\lambda^*} \right).  \label{Eq_23}
\end{equation}
\end{theorem}

\begin{proof}
It suffices to use Theorem~\ref{Thm_12} as well as the relations~(\ref{Eq_17}) and (\ref{Eq_18}).
\end{proof}

\begin{remark}  \rm  \label{Rem_14}
If $ \varphi $ is selected according to (\ref{Eq_14}), the relations~(\ref{Eq_22}) and (\ref{Eq_23}) transform into
\begin{equation}
\lambda^* \, > \, 1, \quad s^* \,\, = \,\, \sum_{i = 1}^n u^*_i \,\, = \,\, \frac{1}{\gamma} \, \ln \, \lambda^*
\label{Eq_24}
\end{equation}
(recall Remark~\ref{Rem_7}).  \qed
\end{remark}

By using the well-known sufficient conditions for stability and instability from the first approximation (see, e.\,g.,
\cite[Theorems~5.1 and 5.2 in Chapter~2]{Godunov1997}), it is easy to establish the following result.

\begin{theorem}  \label{Thm_15}
Let Assumptions~{\rm \ref{Ass_1}, \ref{Ass_4}, \ref{Ass_6},} and {\rm \ref{Ass_9}} hold{\rm ,} and let
$ \: u^* \: = \: (u^*_1, u^*_2, \ldots, u^*_n)^{\top} \: \in \: [0, +\infty)^{n \times 1} \: $ be a nonzero
steady state of the system~{\rm (\ref{Eq_10})}. Suppose also that the growth saturation function~$ \varphi $ is
twice differentiable at the point $ \: s^* \, = \, \sum_{i = 1}^n u^*_i $. Consider the Jacobian matrix
function~$ \mathrm{D} F $ determined by {\rm (\ref{Eq_13})}. At the steady state~$ u^* ${\rm ,} its elements
can be represented as
\begin{equation}
\begin{aligned}
& \frac{\partial F_i(u^*)}{\partial u_j} \:\: = \:\: \begin{cases}
\frac{\varphi'(s^*)}{\varphi(s^*)} \, d_i u^*_i \,\, + \,\, \varphi(s^*) \, g_{ii} \,\, - \,\, d_i, & i = j, \\
\frac{\varphi'(s^*)}{\varphi(s^*)} \, d_i u^*_i \,\, + \,\, \varphi(s^*) \, g_{ij}, & i \neq j,
\end{cases} \\
& i = \overline{1, n}, \quad j = \overline{1, n}.
\end{aligned}  \label{Eq_25}
\end{equation}
If the real parts of all eigenvalues of $ \, \mathrm{D} F(u^*) $ are negative{\rm ,} then the steady state~$ u^* $ is
asymptotically stable. If at least one eigenvalue of $ \mathrm{D} F(u^*) $ has positive real part{\rm ,} then $ u^* $
is unstable.
\end{theorem}

\begin{remark}  \rm  \label{Rem_16}
Recall the boundedness of the total population size along state trajectories as mentioned in Theorem~\ref{Thm_5}.
If a state trajectory is not attracted by the nontrivial steady state with positive components (see Theorem~\ref{Thm_13}),
the trivial steady state at the origin may be approached. The latter case means the extinction of the whole considered
quasispecies population and is possible if all of the death rates are sufficiently large.

Note that the Jacobian matrix $ \mathrm{D} F(0) $ at the origin equals $ \: \varphi(0) \, G \, - \, D $.
An interesting open problem is finding out sufficient conditions for the following dichotomic property which seems natural:
either (i)~the trivial steady state at the origin is asymptotically stable globally in the nonnegative orthant while
the nonzero steady state $ \, u^* \in [0, +\infty)^{n \times 1} \, $ does not exist, or (ii)~the trivial steady state is unstable
while the nontrivial steady state~$ u^* $ exists and is asymptotically stable globally in the nonnegative orthant with
excluded origin.  \qed
\end{remark}

\subsection{Numerical simulations}
\label{Subsect_2_4}

The numerical simulation results in this subsection demonstrate the possibility of two scenarios for our open models.
Namely, with increasing a certain indicative parameter, the error threshold is observed when the population does not
extinct, or, alternatively, the extinction already occurs prior to the error threshold. In the second case, the error
threshold is understood nominally.

For simplicity, we adopt the exponential form~(\ref{Eq_14}) of the growth saturation function~$ \varphi $ and consider
the dynamical system~(\ref{Eq_10}) with the permutation invariant Crow--Kimura formalism such that $ n = N + 1 $,
$ \, G = M + \hat{M}, \, $ and $ \hat{M} $ is given by (\ref{Eq_5}). The selected form of $ \varphi $ allows for using
the simplifications mentioned in Remarks~\ref{Rem_7} and \ref{Rem_14}.

We take
\begin{equation}
\arraycolsep=1.4pt
\def\arraystretch{1.5}
\begin{array}{c}
\gamma = 1, \quad N = 50, \\
m \: = \: (10, 1, 1, \ldots, 1)^{\top} \: \in \: (0, +\infty)^{51 \times 1}, \quad M \: = \: \mathrm{diag} \, m, \\
d \: = \: (\Delta, 0.5, 0.5, \ldots, 0.5)^{\top} \: \in \: (0, +\infty)^{51 \times 1}, \quad D \: = \: \mathrm{diag} \, d,
\end{array}  \label{Eq_25_2}
\end{equation}
while the mutation rate parameter~$ \mu $ (see (\ref{Eq_5})) and the death rate~$ \Delta $ for the genotypes of the zero class
are not fixed, so that the key quantities~$ \lambda^* $ and $ u^* $ become functions of $ \mu $ and $ \Delta $.

Fig.~\ref{Fig_1} illustrates the situation when, with the increase of the mutation rate~$ \mu $, the population clearly goes through
the error threshold with no extinction. It is also demonstrated that, as the death rate~$ \Delta $ of the zero class grows, the error
threshold occurs for lower~$ \mu $. The graphs of the maximum over the real parts of the eigenvalues of $ \mathrm{D} F(u^*) $
yield asymptotic stability of $ u^* $ for all observed values of $ \mu $ except for the critical value of $ \mu $ specifying
the error threshold. However, it is natural to expect the stability for this critical $ \mu $ as well. Besides, extensive numerical
simulations allow for conjecturing that the existence of the nontrivial steady state~$ u^* $ would imply its asymptotic stability globally
in the nonnegative orthant with excluded origin (recall Remark~\ref{Rem_16}).

\begin{figure}[!th]
\includegraphics[width=0.48\textwidth]{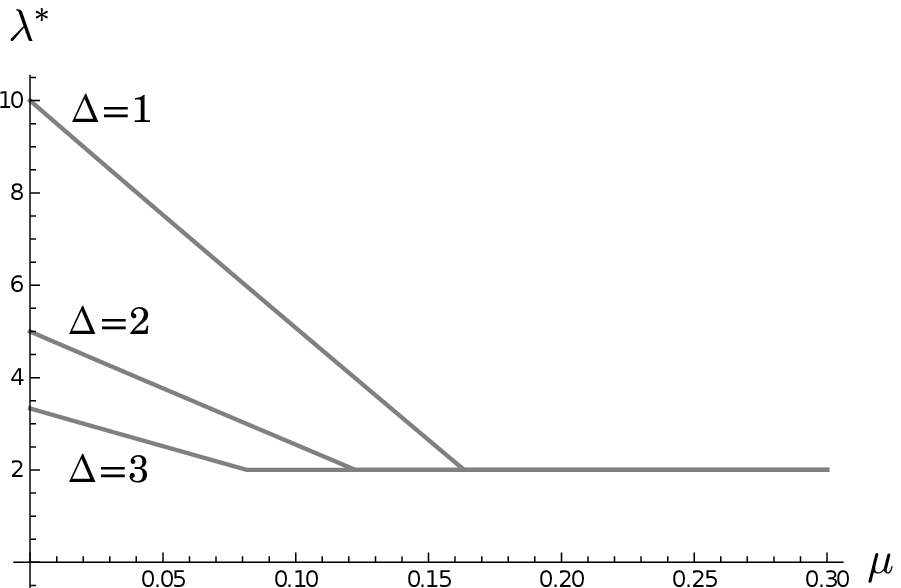} \hfill
\includegraphics[width=0.48\textwidth]{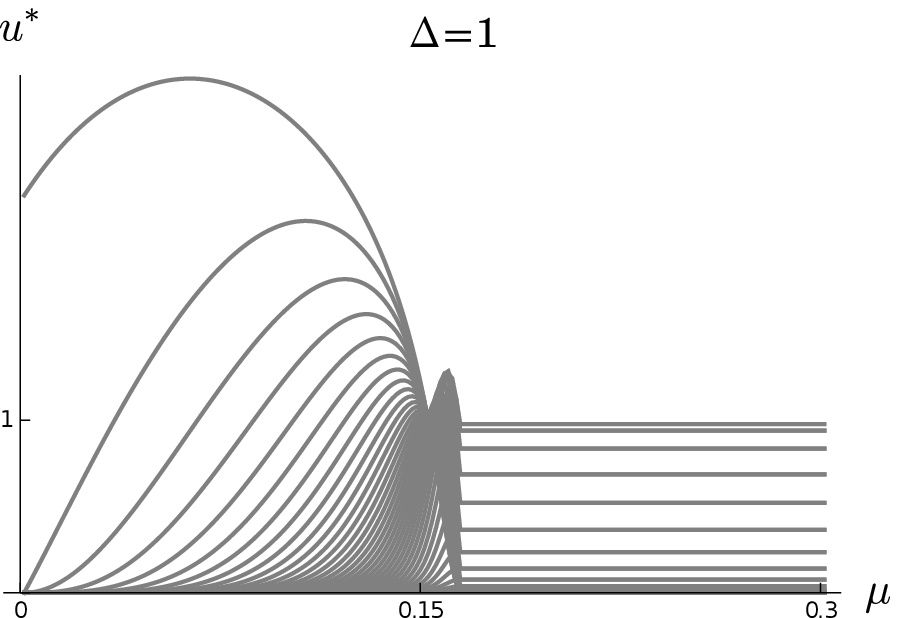} \\

\vspace{2mm}

\includegraphics[width=0.48\textwidth]{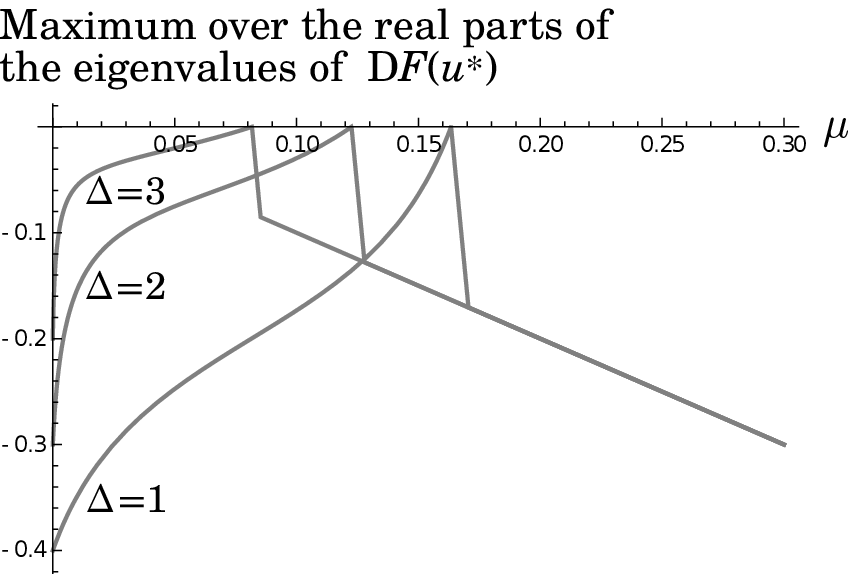} \hfill
\includegraphics[width=0.48\textwidth]{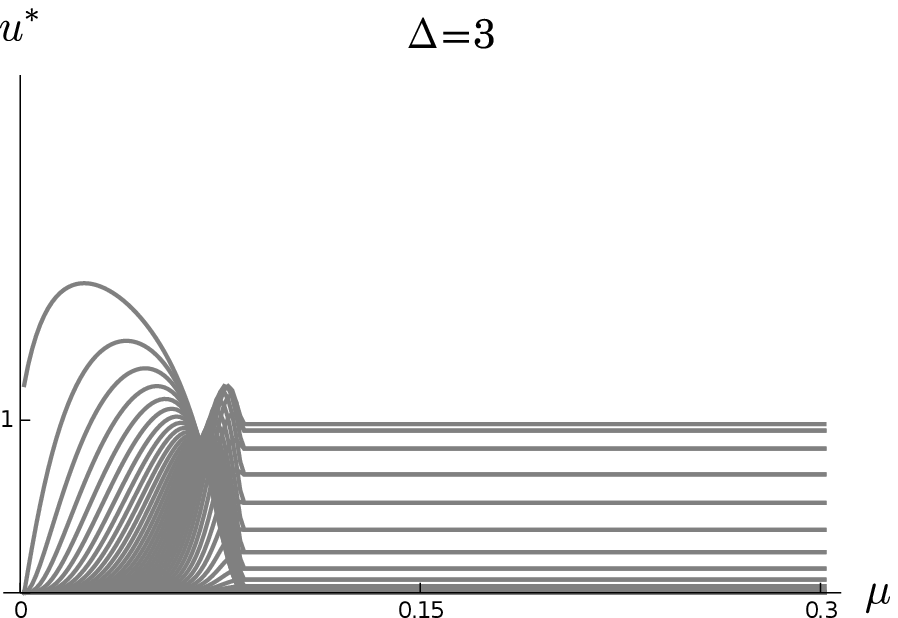}
\bf \caption{\rm Numerical simulation results in Subsection~\ref{Subsect_2_4} for the parameter values~(\ref{Eq_25_2}).
The top left subfigure illustrates the dependence of the steady-state population fitness~$ \, \tilde{f}[u^*] = \lambda^* \, $ on the mutation rate
parameter~$ \mu $; here $ \Delta = 1,2,3 $ for the curves from top to bottom, respectively. The right subfigures show the dependence of
the components of the stationary quasispecies distribution~$ u^* $ on $ \mu $, the top right subfigure corresponds to $ \Delta = 1 $,
the bottom right subfigure relates to $ \Delta = 3 $, and the error threshold with no extinction can be seen in both of them. Moreover,
they demonstrate that, as the death rate~$ \Delta $ of the zero class grows, the error threshold occurs for lower~$ \mu $. Finally,
the bottom left subfigure indicates how the maximum over the real parts of the eigenvalues of the Jacobian matrix~$ \mathrm{D} F(u^*) $
depends on $ \mu $; here $ \Delta = 1,2,3 $ for the curves from bottom to top, respectively. For the bottom left subfigure, Theorem~\ref{Thm_15}
allows asymptotic stability of $ u^* $ to be concluded for all observed values of $ \mu $ except for the critical value of $ \mu $
specifying the error threshold. It is also natural to expect the stability for this critical $ \mu $, even though Theorem~\ref{Thm_15}
cannot be applied in case of the vanishing maximal real part.}
\label{Fig_1}
\end{figure}

Fig.~\ref{Fig_2} shows that the death rate~$ \Delta $ can also be considered as a bifurcation parameter whose changes eventually lead to
the error threshold.

\begin{figure}
\centering
\includegraphics[width=0.48\textwidth]{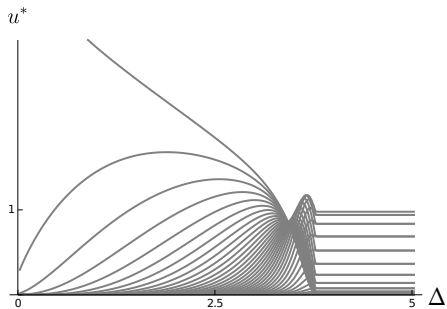}
\bf \caption{\rm Numerical simulation results in Subsection~\ref{Subsect_2_4} for the parameter values~(\ref{Eq_25_2}) and fixed $ \mu = 0.05 $.
Dependence of the components of the stationary quasispecies distribution~$ u^* $ on the death rate parameter~$ \Delta $.}
\label{Fig_2}
\end{figure}

We now change the death rate vector to
\begin{equation}
d \: = \: (\Delta, 2, 2, \ldots, 2)^{\top} \: \in \: (0, +\infty)^{51 \times 1}  \label{Eq_25_3}
\end{equation}
(i.\,e., the death rates of all classes except for the zero class are increased from $ 0.5 $ to $ 2 $), while all other
parameters in (\ref{Eq_25_2}) remain the same. Fig.~\ref{Fig_3} illustrates that the error threshold nominally appears
already after the extinction threshold~$ \lambda^* = 1 $ and hence has almost no effect on the dynamics. The extinction takes place
when $ 0 < \lambda^* \leqslant 1 $; nonpositive values of $ \lambda^* $ are not observed. Note that, according to Theorem~\ref{Thm_13}
and Remark~\ref{Rem_14}, the nontrivial steady state~$ u^* $ exists if and only if $ \lambda^* > 1 $. For $ 0 < \lambda^* \leqslant 1 $,
the right-hand side of the normalization condition $ \: \sum_{i = 1}^n u^*_i \, = \, \ln(\lambda^*) / \gamma \: $
(see (\ref{Eq_24})) becomes negative, and the nontrivial steady state does not exist, so we understand $ u^* $ nominally.

\begin{figure}[!th]
\centering
\includegraphics[width=0.48\textwidth]{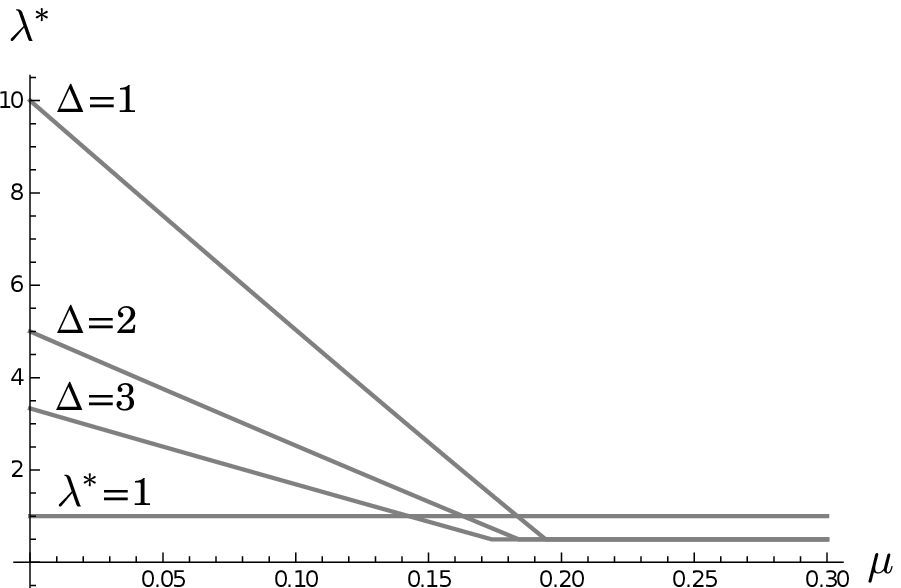} \hfill
\includegraphics[width=0.48\textwidth]{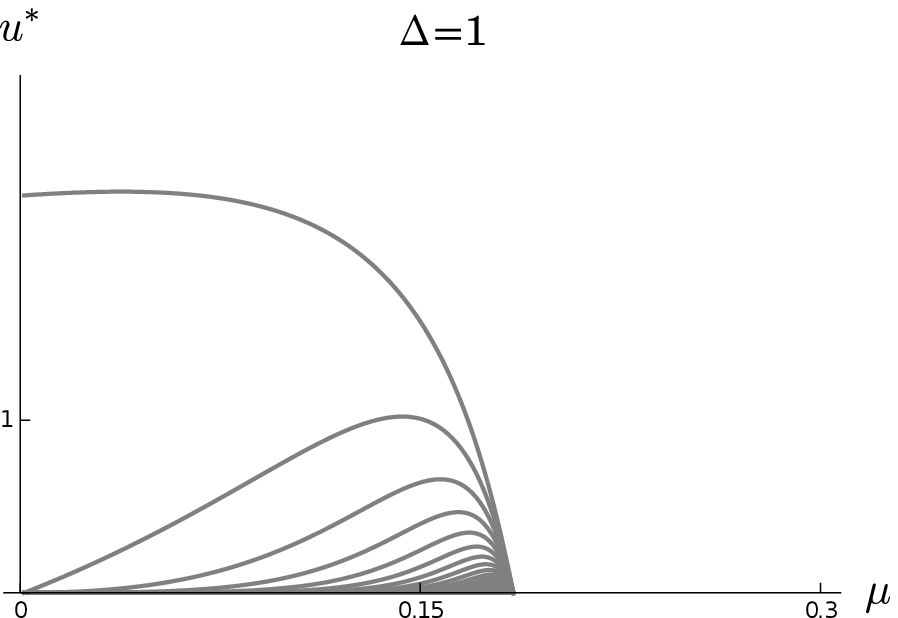} \\

\vspace{4mm}

\includegraphics[width=0.48\textwidth]{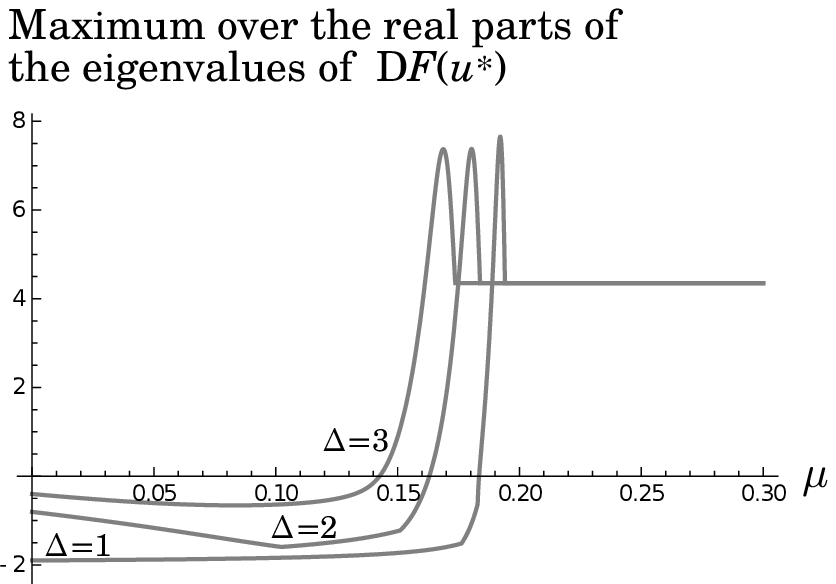}
\bf \caption{\rm Numerical simulation results in Subsection~\ref{Subsect_2_4} for the death rates~(\ref{Eq_25_3}) and all other
parameters given by (\ref{Eq_25_2}). The top left subfigure illustrates the dependence of the steady-state population
fitness~$ \, \tilde{f}[u^*] = \lambda^* \, $ on the mutation rate parameter~$ \mu $, as well as the critical value~$ \lambda^* = 1 $
(due to Theorem~\ref{Thm_13} and Remark~\ref{Rem_14}, the nontrivial steady-state quasispecies distribution exists if and only if
$ \lambda^* > 1 $); here $ \Delta = 1,2,3 $ and $ \lambda^* = 1 $ for the curves from top to bottom, respectively. The top right
subfigure shows the dependence of the components of $ u^* $ on $ \mu $ for the fixed~$ \Delta = 1 $. The extinction occurs when
$ 0 < \lambda^* \leqslant 1 $; nonpositive values of $ \lambda^* $ are not observed. For $ 0 < \lambda^* \leqslant 1 $, the normalization
condition in (\ref{Eq_24}) leads to a negative sum of the components of $ u^* $, so $ u^* $ has a nominal meaning.  The extinction
threshold~$ \lambda^* = 1 $ appears prior to the nominal error threshold. Finally, the bottom subfigure indicates how the maximum
over the real parts of the eigenvalues of the Jacobian matrix~$ \mathrm{D} F(u^*) $ depends on $ \mu $; here $ \Delta = 1,2,3 $ for
the curves from bottom to top, respectively.}
\label{Fig_3}
\end{figure}

The presented numerical simulation results indicate that even the simplified permutation invariant Crow--Kimura model in our open
setting possesses a rich dynamical behavior with different bifurcation scenarios.

\section{Steady-state fitness maximization}
\label{Sec_3}

\subsection{Theoretical analysis}

Similarly to the discussion in \cite[Introduction]{BratusDrozhzhinYakushkina2018}, we adopt the hypotheses that
evolutionary adaptation of the fitness landscape of a quasispecies population is significantly slower than
the internal dynamics of the corresponding state variables (in our case, this dynamics is governed by
the system~(\ref{Eq_10})), and that changes in the fitness landscape are aimed at eventually maximizing the population
fitness (specified by Definition~\ref{Def_8} in our case) under some constraints. Based on these hypotheses as well as on
the relations~(\ref{Eq_20}), (\ref{Eq_21}) and Theorem~\ref{Thm_13}, we arrive at the problem of maximizing the population
fitness
$$
\frac{\sum_{i = 1}^n m_i \, u^*_i[m]}{\sum_{i = 1}^n d_i \, u^*_i[m]} \:\, = \:\, \lambda^*[m]
$$
at the nonzero steady state
$ \: u^*[m] \: = \: (u^*_1[m], \, u^*_2[m], \, \ldots, \, u^*_n[m])^{\top} \: \in \: [0, +\infty)^{n \times 1} \: $
of (\ref{Eq_10}) over fitness landscapes $ \: m \: = \: (m_1, m_2, \ldots, m_n)^{\top} \: \in \: \mathbb{R}^{n \times 1} \: $
subject to the constraints
\begin{equation}
m \, \in \, \Pi,  \label{Eq_26}
\end{equation}
\begin{equation}
\lambda^*[m] \: > \: 0, \quad \frac{1}{\lambda^*[m]} \: \in \: (l_{\varphi}, \varphi(0)),  \label{Eq_27}
\end{equation}
where $ \Pi $ is a bounded subset of $ \mathbb{R}^{n \times 1} $, and the death rates~$ \, d_i $, $ i = \overline{1, n}, \, $
are fixed for simplicity.

As was mentioned in the introduction, we try to keep the connection with the classical quasispecies models as close as reasonably
possible. For our open quasispecies models, it is therefore useful to consider first the simplified case when mutations affect
the fitness landscape (and, consequently, the growth rates), but not the death rates. Modeling the influence of mutations also
on the death rates is a possible subject of future research.

Besides, it is useful to note the paper~\cite{BratusSemenovNovozhilov2018} that contains a relevant discussion on
the applicability of the fitness maximization principle, though as applied to classical replicator systems.

As was discussed in Remark~\ref{Rem_16}, the nontrivial steady state may be unstable for some admissible fitness
landscapes (when the growth terms are overriden by the mortality terms). However, one would expect that, if any optimal
steady state is unstable, all other admissible steady states will also be unstable, although a rigorous derivation of
sufficient conditions for this property remains an open problem. Note also that, in case $ l_{\varphi} = 0 $,
the constraint~(\ref{Eq_27}) is simplified to
\begin{equation}
\lambda^*[m] \: > \: \frac{1}{\varphi(0)} \, .  \label{Eq_28}
\end{equation}
For example, if $ \varphi $ is given by (\ref{Eq_14}), then (\ref{Eq_27}) transforms into $ \, \lambda^*[m] > 1 \, $
(see Remark~\ref{Rem_7}).

Thus, the problem of maximizing the dominant eigenvalue~$ \lambda^*[m] $ of the matrix~$ \, D^{-1} \cdot G[m] \, $
(see Theorem~\ref{Thm_13}), i.\,e.,
\begin{equation}
\lambda^*[m] \: \longrightarrow \: \max_{m \, \in \, \Pi} \, ,  \label{Eq_29}
\end{equation}
plays a central role.

\begin{remark}  \rm  \label{Rem_17}
For the classical quasispecies models described by the system~(\ref{Eq_6}), one has to consider the dominant
eigenvalue of the matrix~$ G $, and, therefore, the subsequent investigation regarding the optimization
problem~(\ref{Eq_29}) will also be valid in the classical case after the formal replacement of $ D $ with
the identity matrix of size~$ n \times n $.  \qed
\end{remark}

\begin{remark}  \rm  \label{Rem_18}
If $ \lambda^*[m] $ is defined for all $ m \in \Pi $, then this is a continuous function of $ m \in \Pi $ due to
the well-known fact that the eigenvalues of a real or complex square matrix continuously depend on its entries
(see, e.\,g., \cite[Theorem~2.11]{Zhang2011}).  \qed
\end{remark}

An important case is when (\ref{Eq_29}) becomes a convex optimization problem.

\begin{assumption}  \label{Ass_19}
$ \Pi $ is a convex compact subset of $ \, \mathbb{R}^{n \times 1} $.
\end{assumption}

\begin{theorem}  \label{Thm_20}
Let Assumptions~{\rm \ref{Ass_1}, \ref{Ass_4}, \ref{Ass_6}, \ref{Ass_9},} and {\rm \ref{Ass_19}} hold{\rm ,} and
let the matrix~$ G $ have the Crow--Kimura form $ \, G = M + \hat{M}, \, $ where $ \, M \, = \, \mathrm{diag} \, m, \, $
while $ \hat{M} $ does not depend on $ m $. Then the function
$$
\mathbb{R}^{n \times 1} \, \ni \, m \:\, \longmapsto \:\, \lambda^*[m] \, \in \, \mathbb{R}
$$
is convex.
\end{theorem}

\begin{proof}
It suffices to use the result of \cite{Cohen1981} saying that the dominant eigenvalue of an essentially nonnegative matrix
is convex if considered as a function of the diagonal elements of this matrix.
\end{proof}

\begin{remark}  \rm  \label{Rem_21}
If $ G $ has the Eigen form $ \, G = K M, \, $ where $ \, M \, = \, \mathrm{diag} \, m, \, $ while $ K $ satisfies
(\ref{Eq_2}) and does not depend on $ m $, then the result of \cite{Cohen1981} cannot be applied directly. In this case,
obtaining sufficient conditions for the convexity of the dominant eigenvalue~$ \lambda^*[m] $ of $ \, D^{-1} \cdot G[m] \, $
with respect to $ m $ remains an open problem.  \qed
\end{remark}

We now assume that the considered system satisfies this convexity property on the fitness landscape constraint set~$ \Pi $.

\begin{assumption}  \label{Ass_22}
The matrix~$ \, D^{-1} \cdot G[m] \, $ admits the dominant eigenvalue $ \, \lambda^*[m] \in \mathbb{R} \, $ for any
$ m \in \Pi ${\rm ,} and the function
\begin{equation}
\Pi \, \ni \, m \:\, \longmapsto \:\, \lambda^*[m] \, \in \, \mathbb{R}  \label{Eq_30}
\end{equation}
is convex.
\end{assumption}

Recall that a point~$ x $ of a convex set~$ C \subseteq \mathbb{R}^n $ is called an extreme point of $ C $
if and only if there is no way to express $ x $ as a convex combination $ \: \alpha y \, + \, (1 - \alpha) z \: $
with $ y \in C $, $ z \in C $, and $ \alpha \in (0, 1) $, except by putting $ y = z = x $ (see, e.\,g.,
\cite[\S 18]{Rockafellar1970}).

\begin{theorem}  \label{Thm_23}
Let Assumptions~{\rm \ref{Ass_1}, \ref{Ass_4}, \ref{Ass_6}, \ref{Ass_9}, \ref{Ass_19},} and {\rm \ref{Ass_22}} hold.
Then the function~{\rm (\ref{Eq_30})} is continuous on $ \Pi $ and attains its maximum over $ \Pi ${\rm ,} and any
related maximizer is an extreme point of $ \Pi $.
\end{theorem}

\begin{proof}
It suffices to take Remark~\ref{Rem_18}, Assumption~\ref{Ass_19}, and \cite[Corollary~32.3.1]{Rockafellar1970}
into account.
\end{proof}

\begin{corollary}  \label{Cor_24}
If the conditions of Theorem~{\rm \ref{Thm_23}} hold and{\rm ,} moreover{\rm ,} $ \Pi $ is a compact convex polytope{\rm ,}
then any maximizer in the problem~{\rm (\ref{Eq_29})} is a vertex of~$ \Pi $.
\end{corollary}

\subsection{Numerical simulations}
\label{Subsect_3_2}

The numerical simulation results in this subsection are purely illustrative and demonstrate a typical behavior of
the maximal steady-state population fitness with the increase of a mutation rate or a death rate.

We adopt the exponential form~(\ref{Eq_14}) of the growth saturation function~$ \varphi $ and consider
the dynamical system~(\ref{Eq_10}) with the following symmetric Crow--Kimura formalism~\cite{BaakeGabriel1999,BratusSemenovNovozhilov2018}:
\begin{itemize}
\item  $ n = 2^N $ genotypes are considered, they are labelled by $ \, 0, 1, \ldots, n - 1 \, $ and associated to
the $ N $-dimensional binary representations (sequences) of the corresponding numbers;
\item  $ G \, = \, M + \hat{M} $, $ \: M \, = \, \mathrm{diag} \, m $,
$ \: \hat{M} \, = \, \{ \hat{m}_{ij} \}_{i,j \, = \, 1}^n $,
\begin{equation}
\begin{aligned}
& \hat{m}_{i + 1, \, j + 1} \:\, = \:\, \begin{cases}
\mu, & H_{ij} = 1, \\
0, & H_{ij} > 1, \\
-N \mu, & H_{ij} = 0 \:\:\: (i = j),
\end{cases} \\
& i \, = \, \overline{0, n - 1}, \quad j \, = \, \overline{0, n - 1},
\end{aligned}  \label{Eq_25_4}
\end{equation}
where $ \mu > 0 $ is a mutation rate parameter and $ H_{ij} $ denotes the Hamming distance between sequences~$ i $ and $ j $
($ H_{ij} = 0 $ only if $ i = j $).
\end{itemize}
The irreducibility of such a mutation matrix~$ \hat{M} $ can be verified via the same reasonings as in the end of Remark~\ref{Rem_11}
for the permutation invariant Crow--Kimura model. Hence, Theorems~\ref{Thm_20}, \ref{Thm_23} and Corollary~\ref{Cor_24} can be
applied, and the steady-state fitness maximization problem for fitness landscapes constrained by a compact convex polytope~$ \Pi $
reduces to sorting through the vertices of $ \Pi $.

Here we do not use the permutation invariant Crow--Kimura formalism (see (\ref{Eq_5})) because of its high heterogeneity such that
many of the original $ 2^N $ binary sequences are concentrated in the middle of the ordered group of the $ N + 1 $ equivalence
classes. Thus, for attenuating the related heterogeneity in the steady-state fitness profile, one typically has to select
a rather specific structure of the death rate vector.

For simplicity, we consider the low-dimensional case
\begin{equation}
N = 3, \quad n = 2^N = 8.  \label{Eq_25_5}
\end{equation}
Moreover, we take
\begin{equation}
\gamma = 1  \label{Eq_25_6}
\end{equation}
as in (\ref{Eq_25_2}). The components of the death rate vector~$ d $ are independently and randomly generated from the uniform
distribution on $ (0, 1) $, and they are also arranged in ascending order:
\begin{equation}
\begin{aligned}
& d \:\, = \:\, (0.099122, \, 0.158373, \, 0.445851, \, 0.453362, \\
& \qquad\quad
0.484529, \, 0.488275, \, 0.580878, \, 0.990911)^{\top} \:\, \in \:\, (0, 1)^{8 \times 1}, \\
& D \: = \: \mathrm{diag} \, d.
\end{aligned}  \label{Eq_25_7}
\end{equation}
The constraint set~$ \Pi $ is chosen as the following simplex with center at the point
$ \: (0.5, 0.5, \ldots, 0.5)^{\top} \, \in \, \mathbb{R}^{n \times 1} $:
\begin{equation}
\begin{aligned}
\Pi \:\: = \:\: \left\{ m \: = \: (m_1, m_2, \ldots, m_n)^{\top} \: \colon \:
{}^{{}^{{}^{{}^{{}^{{}^{{}^{{}^{}}}}}}}}
m_i \, = \, 0.5 + \xi_i, \right. & \\
\left. \xi_i \geqslant 0, \:\:\: i = \overline{1, n}, \:\:\: \sum_{j = 1}^n \xi_j \: = \: 2 \right\}. &
\end{aligned}  \label{Eq_25_8}
\end{equation}
There are $ n = 8 $ vertices which we label by $ \: 1, 2, \ldots, n $: for every $ \: i \, \in \, \{ 1, 2, \ldots, n \}, \: $
the $ i $-th coordinate of vertex~$ i $ is $ 2.5 $, while all other coordinates of vertex~$ i $ are equal to $ 0.5 $.

Fig.~\ref{Fig_4} indicates how the steady-state fitnesses at vertices~$ 1,2,3 $ depend on the mutation rate parameter~$ \mu $.
The maximal steady-state fitness in the problem~(\ref{Eq_29}) is achieved at vertex~1. This is not surprising, since
the components of the death rate vector~(\ref{Eq_25_7}) are written in ascending order. Furthermore, when $ \mu $ is
sufficiently large, the difference between the steady-state fitnesses at the vertices of $ \Pi $ (not only 1--3, but also 4--8)
becomes negligible.

\begin{figure}
\centering
\includegraphics[width=0.48\textwidth]{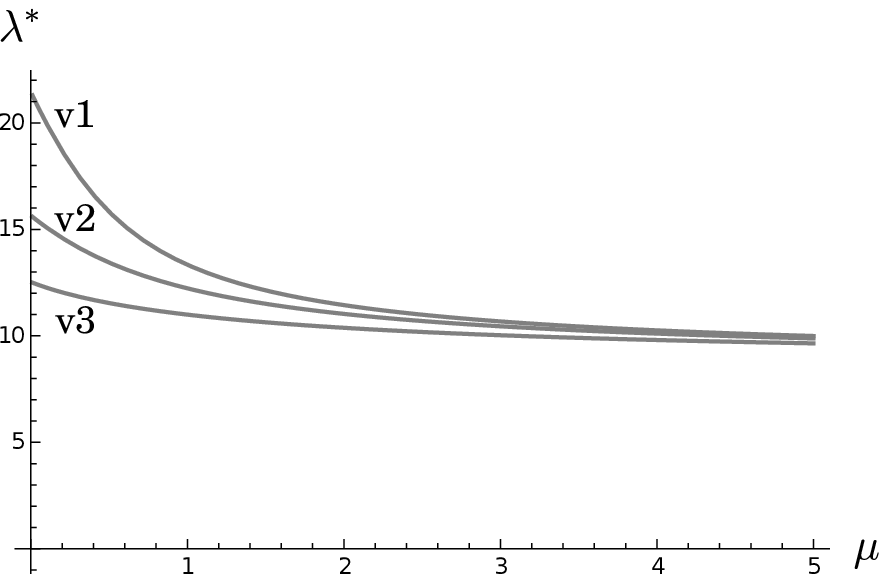}
\bf \caption{\rm Numerical simulation results in Subsection~\ref{Subsect_3_2} for the parameter values~(\ref{Eq_25_4})--(\ref{Eq_25_8}).
The curves from top to bottom correspond to the steady-state population fitnesses~$ \, \tilde{f}[u^*] = \lambda^* \, $ at the first,
second, and third vertices (denoted by v1,\,v2,\,v3) of the simplex~(\ref{Eq_25_8}), respectively. The figure shows the dependence of
these fitnesses on the mutation rate parameter~$ \mu $. The maximum in the problem~(\ref{Eq_29}) is achieved at the first vertex,
which seems natural, since the minimal component of the death rate vector~(\ref{Eq_25_7}) is the first one. For sufficiently large
$ \mu $, the difference between the steady-state fitnesses at the vertices of $ \Pi $ becomes negligible.}
\label{Fig_4}
\end{figure}

Next, we fix $ \mu $ and let the first component of the death rate vector be a variable~$ \Delta $:
\begin{equation}
\begin{aligned}
& \mu = 0.1, \\
& d \:\, = \:\, (\Delta, \, 0.158373, \, 0.445851, \, 0.453362, \\
& \qquad\quad
0.484529, \, 0.488275, \, 0.580878, \, 0.990911)^{\top} \:\, \in \:\, (0, +\infty)^{8 \times 1}.
\end{aligned}  \label{Eq_25_9}
\end{equation}
Fig.~\ref{Fig_5} illustrates how the steady-state fitnesses at vertices~1 and 2 depend on the death rate~$ d_1 = \Delta $ of
genotype~0 (our notation is such that, for every $ \: i \, \in \, \{ 1, 2, \ldots, n \}, \: $ the $ i $-th component of $ d $ is
the death rate of genotype~$ i - 1 $). The maximal steady-state fitness in the problem~(\ref{Eq_29}) is achieved at vertex~1 if
$ \, \Delta \leqslant d_2 = 0.158373 \, $ and at vertex~2 if $ \Delta \geqslant d_2 $.

\begin{figure}
\centering
\includegraphics[width=0.48\textwidth]{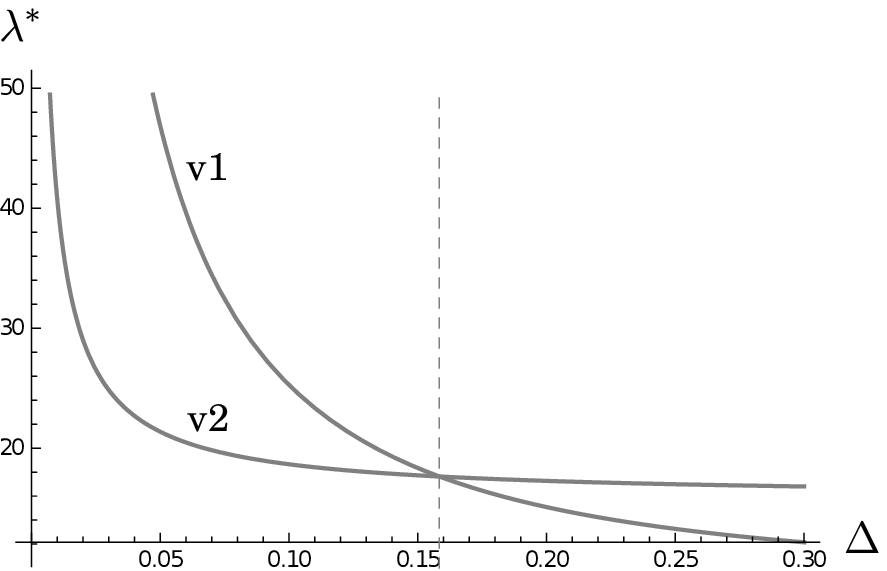}
\bf \caption{\rm Numerical simulation results in Subsection~\ref{Subsect_3_2} for the parameter values~(\ref{Eq_25_4})--(\ref{Eq_25_6}),
(\ref{Eq_25_8}), (\ref{Eq_25_9}). The curves from top to bottom correspond to the steady-state population
fitnesses~$ \, \tilde{f}[u^*] = \lambda^* \, $ at the first and second vertices (denoted by v1,\,v2) of the simplex~(\ref{Eq_25_8}),
respectively. The figure shows the dependence of these fitnesses on the death rate~$ d_1 = \Delta $. The maximum in the problem~(\ref{Eq_29})
is achieved at the first vertex if $ \, \Delta \leqslant d_2 = 0.158373 \, $ and at the second vertex if $ \Delta \geqslant d_2 $.}
\label{Fig_5}
\end{figure}

\section{Continuous-time distributed open quasispecies models}
\label{Sec_4}

\subsection{Constructing the dynamical equation}

Since the actual number of genotypes or alleles involved in a studied process can be extremely large (even though
a significant fraction of them may have small effects), distributed (continuum-of-alleles) quasispecies models
are of particular interest (see, e.\,g., \cite[Introduction to Chapter~IV]{Burger2000}). For the classical Eigen
and Crow--Kimura models, distributed extensions in the form of an integro-differential equation with respect to
a sought-after time-varying probability density were rigorously investigated in
\cite{Burger1986,Burger1988a,Burger1988b,Burger1996} and \cite[Chapter~IV]{Burger2000}, while the first
particular distributed model of this kind was proposed in \cite{CrowKimura1964,Kimura1965}. B\"urger's formalism
treats the distributed Eigen and Crow--Kimura settings together and allows to establish sufficient conditions for
the existence and uniqueness of an equilibrium probability density with its positivity and asymptotic stability
(see \cite[\S 3, \S 4]{Burger1988a} and \cite[\S IV.3]{Burger2000}).

Typical (although not the most general) integro-differential equations describing the Eigen and Crow--Kimura dynamics
can be written as
\begin{equation}
\begin{aligned}
& \frac{\partial p(x, t)}{\partial t} \:\: = \:\: \int\limits_{\Omega} K(x, y) \, m(y) \, p(y, t) \, \mathrm{d} y \:\, - \:\,
f[p(\cdot, t)] \, p(x, t), \\
& x \in \Omega, \quad t \geqslant 0,
\end{aligned}  \label{Eq_31}
\end{equation}
and
\begin{equation}
\begin{aligned}
& \frac{\partial p(x, t)}{\partial t} \:\: = \:\: m(x) \, p(x, t) \:\, + \:\,
\int\limits_{\Omega} K(x, y) \, \mu(y) \, p(y, t) \, \mathrm{d} y \\
& \qquad\qquad\qquad\qquad \:\:
 - \:\, \mu(x) \, p(x, t) \:\, - \:\, f[p(\cdot, t)] \, p(x, t), \\
& x \in \Omega, \quad t \geqslant 0,
\end{aligned}  \label{Eq_32}
\end{equation}
respectively, where the following notations are used:
\begin{itemize}
\item  $ t \geqslant 0 $ is a time variable;
\item  the considered genotypes (or allelic effects) are represented as points of a finite-dimensional region~$ \Omega $;
\item  $ p \, \colon \: \Omega \times [0, +\infty) \: \to \: [0, +\infty) \: $ is the function describing the relative
frequencies of the genotypes and normalized so that
\begin{equation}
\int\limits_{\Omega} p(x, t) \, \mathrm{d} x \,\, = \,\, 1 \quad \forall t \geqslant 0;  \label{Eq_33}
\end{equation}
\item  $ \Omega \, \ni \, x \: \longmapsto \: K(x, y) \, \in \, [0, +\infty) \: $ is the probability density related to
reproduction events such that an individual of genotype~$ y \in \Omega $ produces an individual of genotype~$ x \in \Omega $,
i.\,e., $ \: K \colon \, \Omega^2 \to [0, +\infty) \: $ is a distributed analog of the mutation matrix, and
\begin{equation}
\int\limits_{\Omega} K(x, y) \, \mathrm{d} x \,\, = \,\, 1 \quad \forall y \in \Omega;  \label{Eq_34}
\end{equation}
\item  $ m \colon \, \Omega \to \mathbb{R} \: $ is the distributed fitness landscape;
\item  $ f[p(\cdot, t)] \: = \: \int_{\Omega} m(x) \, p(x, t) \, \mathrm{d} x \: $ is the mean population fitness;
\item  $ \mu \colon \, \Omega \, \to \, [0, +\infty) \: $ is the mutation rate function.
\end{itemize}

For simplicity, we impose some more specific conditions than in \cite[\S 3, \S 4]{Burger1988a} and \cite[\S IV.3]{Burger2000}.
In particular, it is reasonable to consider a bounded region~$ \Omega $ from a computational perspective.

\begin{assumption}  \label{Ass_25}
$ \Omega $ is a bounded open domain in $ \mathbb{R}^{\varkappa} $ or the closure of such a domain{\rm ,}
$ \varkappa \in \mathbb{N} ${\rm ,} the functions $ \: K \colon \, \Omega^2 \to [0, +\infty) ${\rm ,}
$ \: m \colon \, \Omega \to \mathbb{R}, \: $ and $ \: \mu \colon \, \Omega \, \to \, [0, +\infty) \: $ are bounded{\rm ,}
and {\rm (\ref{Eq_34})} holds.
\end{assumption}

One can consider (\ref{Eq_31}) and (\ref{Eq_32}) as differential equations whose state space is the Banach space
\begin{equation}
V \, = \, L^1(\Omega; \mathbb{R})  \label{Eq_35}
\end{equation}
of all Lebesgue measurable functions $ \: v \colon \, \Omega \to \mathbb{R} \: $ for which the Lebesgue integral
$ \: \int_{\Omega} |v(x)| \, \mathrm{d} x \: $ exists and is finite (the theory of differential equations in Banach spaces is
introduced, e.\,g., in \cite{DaleckiiKrein1974}).

The relations~$ \, p(x, t) \geqslant 0 \, $ and (\ref{Eq_33}) are preserved along the solutions of (\ref{Eq_31}) and (\ref{Eq_32}).

Both equations~(\ref{Eq_31}) and (\ref{Eq_32}) can be written in the common form
\begin{equation}
\begin{aligned}
& \frac{\partial p(x, t)}{\partial t} \:\: = \:\: \int\limits_{\Omega} G(x, y) \, p(y, t) \, \mathrm{d} y \:\, + \:\,
b(x) \, p(x, t) \:\, - \:\, f[p(\cdot, t)] \, p(x, t), \\
& x \in \Omega, \quad t \geqslant 0,
\end{aligned}  \label{Eq_36}
\end{equation}
with the functions $ \: G \colon \, \Omega^2 \to \mathbb{R} \: $ and $ \: b \colon \, \Omega \to \mathbb{R} \: $ given by
\begin{equation}
\begin{aligned}
G(x, y) \:\, = \:\, K(x, y) \, m(y) \quad \mbox{and} \quad b(x) \:\, = \:\, 0 & \\
\mbox{for the Eigen model~(\ref{Eq_31}),} &
\end{aligned}  \label{Eq_37}
\end{equation}
or
\begin{equation}
\begin{aligned}
G(x, y) \:\, = \:\, K(x, y) \, \mu(y) \quad \mbox{and} \quad b(x) \:\, = \:\, m(x) \, - \, \mu(x) & \\
\mbox{for the Crow--Kimura model~(\ref{Eq_32}).} &
\end{aligned}  \label{Eq_38}
\end{equation}

Note that a solution of (\ref{Eq_36}) with $ \: \int_{\Omega} p(x, 0) \, \mathrm{d} x \, = \, 1 \: $ can be represented as
$$
p(x, t) \:\, = \:\, \frac{r(x, t)}{\int_{\Omega} r(y, t) \, \mathrm{d} y} \quad
\forall x \in \Omega \quad \forall t \geqslant 0,
$$
where
$$
r(x, t) \:\, = \:\, c \: \exp \left\{ \int\limits_0^t f[p(\cdot, \tau)] \, \mathrm{d} \tau \right\} \, p(x, t) \quad
\forall x \in \Omega \quad \forall t \geqslant 0
$$
is the unique solution of
\begin{equation}
\frac{\partial r(x, t)}{\partial t} \:\: = \:\: \int\limits_{\Omega} G(x, y) \, r(y, t) \, \mathrm{d} y \:\, + \:\,
b(x) \, r(x, t), \quad x \in \Omega, \quad t \geqslant 0,  \label{Eq_39}
\end{equation}
$$
r(x, 0) \: = \: c \: p(x, 0), \quad x \in \Omega,
$$
with an arbitrary constant $ \: c \, \in \, \mathbb{R} \setminus \{ 0 \}, \: $ i.\,e., the values of $ r $ may be
interpreted as some quantities whose normalization leads to the values of the time-varying probability density~$ p $.

We therefore proceed from the system~(\ref{Eq_39}) and incorporate growth and mortality characteristics in it
(recall similar considerations in Subsection~\ref{Subsect_2_2}). Let
$ \: u \, \colon \: \Omega \times [0, +\infty) \: \to \: [0, +\infty) \: $ be a function specifying the sought-after
dynamical quantities associated with the considered genotypes. Introduce also the death rate function
$ \: d \colon \, \Omega \, \to \, [0, +\infty) \: $ and the growth saturation term
$ \: \varphi \left( \int_{\Omega} u(x, t) \, \mathrm{d} x \right) $.

\begin{assumption}  \label{Ass_26}
$ d \colon \, \Omega \, \to \, [0, +\infty) \: $ is bounded{\rm ,} $ \: \varphi \colon \, \mathbb{R} \to [0, +\infty) \: $
fulfills the same properties as mentioned in Assumption~{\rm \ref{Ass_1},} $ \: G \colon \, \Omega^2 \to \mathbb{R} \: $ and
$ \: b \colon \, \Omega \to \mathbb{R} \: $ are defined by either {\rm (\ref{Eq_37})} or {\rm (\ref{Eq_38}),} and{\rm ,}
moreover{\rm ,}
$$
G(x, y) \, \geqslant \, 0 \quad \forall x \in \Omega \quad \forall y \in \Omega.
$$
\end{assumption}

Thus, we arrive at the integro-differential equation
\begin{equation}
\begin{aligned}
& \frac{\partial u(x, t)}{\partial t} \:\: = \:\: \varphi \left( \int\limits_{\Omega} u(y, t) \, \mathrm{d} y \right) \,
\left( \int\limits_{\Omega} G(x, y) \, u(y, t) \, \mathrm{d} y \:\, + \:\, b(x) \, u(x, t) \right) \\
& \qquad\qquad\quad \:\:
- \:\, d(x) \, u(x, t), \\
& x \in \Omega, \quad t \geqslant 0.
\end{aligned}  \label{Eq_40}
\end{equation}
Consider it as a differential equation with the state space~(\ref{Eq_35}), and set the initial condition as
\begin{equation}
u(\cdot, 0) \: = \: u_0(\cdot) \: \in \: V_+,  \label{Eq_41}
\end{equation}
where
\begin{equation}
\begin{aligned}
V_+ \:\: = \:\: \{ v \in V \: \colon \:\, & \mbox{$ v(x) \geqslant 0 \, $ for almost all $ x \in \Omega $} \\
& \mbox{with respect to Lebesgue measure in $ \, \mathbb{R}^{\varkappa} \supset \Omega $} \}.
\end{aligned}  \label{Eq_42}
\end{equation}
The zero element in $ V $ is a trivial steady state of (\ref{Eq_40}).

For any real Banach spaces~$ W_1 $ and $ W_2 $, let $ \mathcal{L}(W_1, W_2) $ denote the set of all bounded linear
operators acting from $ W_1 $ to $ W_2 $.

It is convenient to rewrite the right-hand side of (\ref{Eq_40}) by means of the following operator notation:
\begin{equation}
\begin{aligned}
& \mathcal{G} \, \in \, \mathcal{L}(V, V), \quad \psi \, \in \, \mathcal{L}(V, \mathbb{R}), \\
& \mathcal{B} \, \in \, \mathcal{L}(V, V), \quad \mathcal{D} \, \in \, \mathcal{L}(V, V), \quad
\mathcal{F} \, \in \, \mathcal{L}(V, V), \\
& \mathcal{G}[v](x) \:\, = \:\, \int\limits_{\Omega} G(x, y) \, v(y) \, \mathrm{d} y, \quad
\psi[v] \:\, = \:\, \int\limits_{\Omega} v(y) \, \mathrm{d} y, \\
& \mathcal{B}[v](x) \:\, = \:\, b(x) \, v(x), \quad \mathcal{D}[v](x) \:\, = \:\, d(x) \, v(x), \\
& \mathcal{F}[v](x) \:\: = \:\: \varphi(\psi[v]) \: (\mathcal{G}[v](x) \: + \: \mathcal{B}[v](x)) \,\, - \,\,
\mathcal{D}[v](x) \\
& \forall x \in \Omega \quad \forall v \in V.
\end{aligned}  \label{Eq_43}
\end{equation}

\begin{theorem}  \label{Thm_27}
Under Assumptions~{\rm \ref{Ass_25}} and {\rm \ref{Ass_26},} the Fr\'echet derivative (Jacobian operator)
$ \: \mathrm{D} \mathcal{F} \, \colon \: V \, \to \, \mathcal{L}(V, V) \: $ of the right-hand side
operator~$ \mathcal{F} $ is represented as
\begin{equation}
\mathrm{D} \mathcal{F} [v] \:\, = \:\, \varphi(\psi[v]) \: (\mathcal{G} \, + \, \mathcal{B}) \,\, - \,\,
\mathcal{D} \,\, + \,\, \mathcal{E}_v \quad \forall v \in V,  \label{Eq_44}
\end{equation}
where $ \, \mathcal{E}_v \, \in \, \mathcal{L}(V, V) \, $ is defined by
\begin{equation}
\mathcal{E}_v[h] \:\, = \:\, \varphi'(\psi[v]) \,\, \psi[h] \,\, (\mathcal{G}[v] \, + \, \mathcal{B}[v]) \quad
\forall h \in V  \label{Eq_45}
\end{equation}
for any $ v \in V $.
\end{theorem}

\begin{theorem}  \label{Thm_28}
Let Assumptions~{\rm \ref{Ass_25}} and {\rm \ref{Ass_26}} hold. Then the subset~$ V_+ $ of the state space~$ V $ is
positively invariant with respect to the dynamical system~{\rm (\ref{Eq_40})}. Moreover{\rm ,} for any initial
state~{\rm (\ref{Eq_41}),} there exists a unique solution of {\rm (\ref{Eq_40})} that is defined for all
$ t \geqslant 0 $.
\end{theorem}

Theorems~\ref{Thm_27} and \ref{Thm_28} are proved in Appendix.

The next condition is imposed in particular to guarantee the boundedness of the total population size function
\begin{equation}
s(t) \,\, = \,\, \int\limits_{\Omega} u(x, t) \, dx \quad \forall t \geqslant 0  \label{Eq_46}
\end{equation}
along the solution of (\ref{Eq_40}),~(\ref{Eq_41}).

\begin{assumption}  \label{Ass_29}
There exists a constant~$ d_{\mathrm{low}} > 0 $ such that $ \, d(x) \geqslant d_{\mathrm{low}} \, $ for all
$ x \in \Omega $.
\end{assumption}

\begin{theorem}  \label{Thm_30}
Under Assumptions~{\rm \ref{Ass_25}, \ref{Ass_26},} and {\rm \ref{Ass_29},} the total population size
function~{\rm (\ref{Eq_46})} is bounded along the solution of {\rm (\ref{Eq_40}),~(\ref{Eq_41})}.
\end{theorem}

Theorem~\ref{Thm_30} can be proved similarly to Theorem~\ref{Thm_5}. Compared to the proof of the latter (given in Appendix),
sums over indices $ \, 1, 2, \ldots, n \, $ should now be replaced with integrals over $ \Omega $, while $ m_{\max} $ and
$ d_{\min} $ should be replaced with $ \: \mathrm{ess} \, \sup_{x \, \in \, \Omega} \, m(x) \: $ and $ d_{\mathrm{low}} $,
respectively.

\subsection{Steady-state analysis}

For investigating a nontrivial steady state of the integro-differential equation~(\ref{Eq_40}), additional conditions
have to be adopted.

\begin{assumption}  \label{Ass_31}
The following properties hold{\rm :}
\begin{list}{\rm \arabic{count})}%
{\usecounter{count}}
\item  the function~$ \varphi $ satisfies the same conditions as mentioned in Assumption~{\rm \ref{Ass_6};}
\item  for every measurable set $ \Omega_1 \subset \Omega $ such that the Lebesgue measures of both
$ \Omega_1 $ and $ \Omega \setminus \Omega_1 $ are positive{\rm ,} one has
$$
\int\limits_{\Omega \, \setminus \, \Omega_1} \, \int\limits_{\Omega_1} \frac{G(x, y)}{d(x)} \: \mathrm{d} x \, \mathrm{d} y
\:\, > \:\, 0;
$$
\item  in the Eigen case~{\rm (\ref{Eq_37}),} there exists a subset~$ \hat{\Omega} \subseteq \Omega $ of positive
Lebesgue measure and a constant~$ \hat{c} > 0 $ such that
\begin{equation}
\frac{G(x, y)}{d(x)} \: \geqslant \: \hat{c} \quad \forall \: (x, y) \, \in \, \hat{\Omega}^2;  \label{Eq_47}
\end{equation}
\item  in the Crow--Kimura case~{\rm (\ref{Eq_38}),} there exists a subset~$ \hat{\Omega} \subseteq \Omega $ of positive
Lebesgue measure and a constant~$ \hat{c} > 0 $ such that {\rm (\ref{Eq_47})} holds together with
$$
\eta \:\, = \:\, \mathrm{ess} \, \sup_{x \, \in \, \Omega} \, \frac{m(x) \, - \, \mu(x)}{d(x)} \:\, = \:\,
\mathrm{ess} \, \sup_{x \, \in \, \hat{\Omega}} \, \frac{m(x) \, - \, \mu(x)}{d(x)}
$$
and
\begin{equation}
\hat{c} \,\, \int\limits_{\hat{\Omega}} \frac{\mathrm{d} x}{\eta \: - \: \frac{m(x) \, - \, \mu(x)}{d(x)}} \:\, > \:\, 1,
\label{Eq_48}
\end{equation}
where divergence of the integral is allowed.
\end{list}
\end{assumption}

\begin{remark}  \rm \label{Rem_32}
Let Assumptions~\ref{Ass_25}, \ref{Ass_26}, and \ref{Ass_29} hold. Then Item~2 of Assumption~\ref{Ass_31} holds if, e.\,g.,
$ G $ is positive on $ \Omega^2 $. Item~2 means that the compact operator $ \, \mathcal{G}_1 \, \in \, \mathcal{L}(V, V) \, $
given by
\begin{equation}
\begin{aligned}
& \mathcal{G}_1 \: = \: \mathcal{D}^{-1} \mathcal{G}, \\
& \mathcal{G}_1[v](x) \:\, = \:\, \frac{1}{d(x)} \: \int\limits_{\Omega} G(x, y) \, v(y) \, \mathrm{d} y \quad
\forall x \in \Omega \quad \forall v \in V
\end{aligned}  \label{Eq_49}
\end{equation}
is irreducible (see \cite[Proposition~3.1]{Burger1988a}). Regarding Item~4, the integral in (\ref{Eq_48}) diverges if, e.\,g.,
there exist a point~$ \hat{x} \in \hat{\Omega} $ and a constant~$ c_1 > 0 $ such that
$$
\frac{m(x) \, - \, \mu(x)}{d(x)} \:\, \geqslant \:\, \eta \: - \: c_1 \, \| x - \hat{x} \| \quad \forall x \in \hat{\Omega}.
$$
\qed
\end{remark}

A nontrivial steady state~$ u^* \in V_+ $ of the integro-differential equation~(\ref{Eq_40}) satisfies
\begin{equation}
\mathcal{G}_1 [u^*] \: + \: \mathcal{B}_1 [u^*] \:\, = \:\, \frac{1}{\varphi(s^*)} \, u^*, \quad
s^* \: = \: \int\limits_{\Omega} u^*(x) \, \mathrm{d} x \: > \: 0,  \label{Eq_50}
\end{equation}
where the operator notations~(\ref{Eq_49}) and
\begin{equation}
\begin{aligned}
& \mathcal{B}_1 \: = \: \mathcal{D}^{-1} \mathcal{B}, \\
& \mathcal{B}_1[v](x) \,\, = \,\, \frac{b(x)}{d(x)} \: v(x) \quad \forall x \in \Omega \quad \forall v \in V
\end{aligned}  \label{Eq_51}
\end{equation}
are used. We hence arrive at the problem of finding a positive real eigenvalue~$ \lambda^* $ of
the operator $ \: \mathcal{G}_1 + \mathcal{B}_1 \, = \, \mathcal{D}^{-1} (\mathcal{G} + \mathcal{B}) \: $
with a related eigenfunction~$ u^* \in V_+ $ such that
\begin{equation}
\varphi(s^*) \: = \: \frac{1}{\lambda^*} \: \in \: (l_{\varphi}, \varphi(0)), \quad
s^* \: = \: \int\limits_{\Omega} u^*(x) \, \mathrm{d} x \: = \: \varphi^{-1} \left( \frac{1}{\lambda^*} \right) \: > \: 0
\label{Eq_52}
\end{equation}
(recall similar reasonings in Subsection~\ref{Subsect_2_3} with the formulas~(\ref{Eq_17}) and (\ref{Eq_18})).
For such $ \lambda^* $ and $ u^* $, the relations $ \: \mathcal{G} [u^*] \, = \, \lambda^* \, \mathcal{D} [u^*] \: $
and (\ref{Eq_34}), (\ref{Eq_37}), (\ref{Eq_38}), (\ref{Eq_43}) lead to
\begin{equation}
\lambda^* \: = \: \frac{\int_{\Omega} m(x) \, u^*(x) \, \mathrm{d} x}{\int_{\Omega} d(x) \, u^*(x) \, \mathrm{d} x}
\label{Eq_53}
\end{equation}
(recall also (\ref{Eq_19})).

Similarly to Definition~\ref{Def_8}, we specify the fitness function as follows.

\begin{definition}  \label{Def_33}
The population fitness for the modified distributed quasispecies models described by the integro-differential
equation~{\rm (\ref{Eq_40})} is given by
\begin{equation}
\begin{aligned}
& \tilde{f}[v] \:\: = \:\: \begin{cases}
0, & \int_{\Omega} v(x) \, \mathrm{d} x \: = \: 0, \\
\frac{f[v]}{\int_{\Omega} d(x) \, v(x) \, \mathrm{d} x} \,\, = \,\,
\frac{\int_{\Omega} m(x) \, v(x) \, \mathrm{d} x}{\int_{\Omega} d(x) \, v(x) \, \mathrm{d} x} \, , &
\int_{\Omega} v(x) \, \mathrm{d} x \: > \: 0,
\end{cases} \\
& \forall v \in V_+.
\end{aligned}
\end{equation}  \label{Eq_54}
\end{definition}

The relation~(\ref{Eq_53}) therefore means that
\begin{equation}
\lambda^* \, = \, \tilde{f} [u^*].  \label{Eq_55}
\end{equation}

\begin{theorem}  \label{Thm_34}
Under Assumptions~{\rm \ref{Ass_25}, \ref{Ass_26}, \ref{Ass_29},} and {\rm \ref{Ass_31},} the following properties hold{\rm :}
\begin{itemize}
\item  the operator $ \: \mathcal{G}_1 + \mathcal{B}_1 \, = \, \mathcal{D}^{-1} (\mathcal{G} + \mathcal{B}) \: $
has a dominant eigenvalue~$ \lambda^* ${\rm ,} which is real and greater than the real part of any other spectral point of
$ \mathcal{G}_1 + \mathcal{B}_1 ${\rm ;}
\item  the eigenvalue~$ \lambda^* $ is simple and admits a positive eigenfunction{\rm ;}
\item  there are no other eigenvalues of $ \mathcal{G}_1 + \mathcal{B}_1 $ admitting nonnegative eigenfunctions.
\end{itemize}
\end{theorem}

Theorem~\ref{Thm_34} is proved in Appendix.

\begin{theorem}  \label{Thm_35}
Let Assumptions~{\rm \ref{Ass_25}, \ref{Ass_26}, \ref{Ass_29},} and {\rm \ref{Ass_31}} hold. Then a steady
state~$ u^* \in V_+ $ of the integro-differential equation~{\rm (\ref{Eq_40})} with
$ \: s^* \: = \: \int_{\Omega} u^*(x) \, \mathrm{d} x \: > \: 0 \: $ exists if and only if the dominant
eigenvalue~$ \lambda^* $ of the operator~$ \mathcal{G}_1 + \mathcal{B}_1 $ satisfies
\begin{equation}
\lambda^* \, > \, 0, \quad \frac{1}{\lambda^*} \: \in \: (l_{\varphi}, \varphi(0)).  \label{Eq_56}
\end{equation}
Moreover{\rm ,} if {\rm (\ref{Eq_56})} holds{\rm ,} this steady state is uniquely determined as the positive
eigenfunction of $ \mathcal{G}_1 + \mathcal{B}_1 $ corresponding to $ \lambda^* $ and normalized so that
\begin{equation}
s^* \,\, = \,\, \int\limits_{\Omega} u^*(x) \, \mathrm{d} x \,\, = \,\, \varphi^{-1} \left( \frac{1}{\lambda^*} \right).
\label{Eq_57}
\end{equation}
\end{theorem}

\begin{proof}
It suffices to use Theorem~\ref{Thm_34} as well as the relations~(\ref{Eq_50}) and (\ref{Eq_52}).
\end{proof}

Note the similarity between Theorems~\ref{Thm_34}, \ref{Thm_35} and \ref{Thm_12}, \ref{Thm_13}, respectively.

\begin{remark}  \rm  \label{Rem_36}
If $ \varphi $ is given by (\ref{Eq_14}), the relations~(\ref{Eq_56}) and (\ref{Eq_57}) transform into
\begin{equation}
\lambda^* \, > \, 1, \quad s^* \,\, = \,\, \int\limits_{\Omega} u^*(x) \, \mathrm{d} x \,\, = \,\,
\frac{1}{\gamma} \, \ln \, \lambda^*  \label{Eq_58}
\end{equation}
(recall Remark~\ref{Rem_14}).  \qed
\end{remark}

Next, let us provide sufficient conditions for stability and instability of the nontrivial steady state~$ u^* $ of
(\ref{Eq_40}). By using the general results mentioned in \cite[\S VII.2.4]{DaleckiiKrein1974}, one obtains
the following theorem.

\begin{theorem}  \label{Thm_37}
Let Assumptions~{\rm \ref{Ass_25}, \ref{Ass_26}, \ref{Ass_29},} and {\rm \ref{Ass_31}} hold{\rm ,} and let
$ u^* \in V_+ $ satisfy $ \: s^* \: = \: \int_{\Omega} u^*(x) \, \mathrm{d} x \: > \: 0 \: $ and be a steady
state of the integro-differential equation~{\rm (\ref{Eq_40})}. Suppose also that the growth saturation
function~$ \varphi $ is twice differentiable at the point~$ s^* $. Consider the Jacobian operator~$ \mathrm{D} \mathcal{F} $
determined by {\rm (\ref{Eq_44})} and {\rm (\ref{Eq_45})}. The auxiliary operator~$ \mathcal{E}_{u^*} $ for
the steady state~$ u^* $ is simplified to
\begin{equation}
\mathcal{E}_{u^*}[h] \:\, = \:\, \frac{\varphi'(\psi[u^*])}{\varphi(\psi[u^*])} \: \psi[h] \: \mathcal{D} [u^*]
\quad \forall h \in V.  \label{Eq_59}
\end{equation}
If the real parts of all spectral points of $ \mathrm{D} \mathcal{F} [u^*] $ are negative{\rm ,} then the steady
state~$ u^* $ is asymptotically stable. If at least one spectral point of $ \mathrm{D} \mathcal{F} [u^*] $ has positive
real part{\rm ,} then $ u^* $ is unstable.
\end{theorem}

\begin{remark}  \rm  \label{Rem_38}
If $ \varphi $ is selected in line with (\ref{Eq_14}), the relation~(\ref{Eq_59}) transforms into
$$
\mathcal{E}_{u^*}[h] \,\, = \,\, -\gamma \: \psi[h] \: \mathcal{D} [u^*] \quad \forall h \in V.
$$
\qed
\end{remark}

\begin{remark}  \rm  \label{Rem_39}
Recall the boundedness of the total population size along state trajectories as mentioned in Theorem~\ref{Thm_30}.
If a state trajectory is not attracted by the positive steady state, the zero steady state may be approached.
The latter case means the extinction of the whole considered quasispecies population and is possible when
$ \: \mathrm{ess} \, \inf_{x \, \in \, \Omega} \, d(x) \: $ is sufficiently large (see also Remark~\ref{Rem_16}).  \qed
\end{remark}

The challenging practical problems of obtaining the dominant eigenvalue~$ \lambda^* $ and the positive eigenfunction~$ u^* $
of the operator~$ \mathcal{G}_1 + \mathcal{B}_1 $ and verifying stability or instability of the steady state~$ u^* $ via
Theorem~\ref{Thm_37} (so that one has to characterize the spectrum of $ \mathrm{D} \mathcal{F} [u^*] $) are possible
subjects of future research. Their principal issue is that eigenvalue problems for noncompact and non-self-adjoint compact
operators are involved. In general, the compact integral operator~$ \mathcal{G}_1 $ is non-self-adjoint (due to its
nonsymmetric kernel), the operator~$ \mathcal{G}_1 + \mathcal{B}_1 $ is compact in the Eigen case~(\ref{Eq_37}) (when
$ \mathcal{B}_1 $ vanishes) but noncompact in the Crow--Kimura case~(\ref{Eq_38}), and $ \mathrm{D} \mathcal{F} [u^*] $ is
noncompact in both of these cases (since $ \mathcal{D} $ is noncompact). After an efficient approach to treating such
problems is developed (at least for a particular nontrivial subclass of the distributed quasispecies models), it will be
reasonable to investigate the infinite-dimensional problems of maximizing the steady-state fitness over fitness landscape
functions that satisfy appropriate constraints.

\section{Conclusion}
\label{Sec_5}

In this work, we developed a general approach to the construction of open quasispecies models incorporating growth and mortality
characteristics. We proposed open modifications of the Eigen and Crow--Kimura quasispecies models in both ODE-based and distributed
(continuum-of-alleles) settings. The distributed formulations built on integro-differential equations and were motivated by
the complexity of numerical analysis of quasispecies systems with large numbers of ODEs, as well as by the fact that the actual
number of genotypes or alleles involved in a studied process could indeed be extremely large. Essential properties of the open
quasispecies models, regarding in particular steady states, were investigated.

We also explained a motivation for steady-state fitness maximization problems and studied them in case of ODE-based quasispecies
dynamics. It was in particular established that, for the ODE-based Crow--Kimura models (both open and closed), such a problem
leads to convex optimization and allows for an efficient numerical implementation. For the Eigen models, it remains an open
question whether a similar reduction to convex optimization can in general be carried out.

Another open problem for our models is verification of the following natural conjecture under the already adopted and possibly
some additional assumptions: either (i) the trivial zero steady state is asymptotically stable globally in the nonnegative
subspace while the nontrivial nonnegative steady state does not exist, or (ii) the trivial steady state is unstable while
the nontrivial steady state exists and is asymptotically stable globally in the nonnegative subspace with excluded zero.

As was demonstrated in our numerical simulation results, the ODE-based open quasispecies models enable at least the following
two nontrivial bifurcation scenarios with the increase of a specific parameter (such as a mutation rate or a death rate):
(i) the error threshold is observed when the population does not extinct, or (ii) the extinction already takes place prior to
the nominally interpreted error threshold.

Our mathematical constructions for the distributed open quasispecies models serve just as the first step in this research
direction. Further developments require the design of an efficient practical approach to treating the related eigenvalue
problems and verifying stability or instability of the nontrivial nonnegative steady states. A major difficulty in these
problems is that one has to deal with noncompact and non-self-adjoint compact operators. After the corresponding challenges
are overcome at least for a particular reasonable subclass of the distributed open quasispecies models, it will be relevant
to consider the infinite-dimensional problem of steady-state fitness maximization.

As was also noted above, a promising research area where our framework may eventually be used is modeling the dynamics of
various quasispecies related to certain pathogens or diseased cells under targeted attacks from therapeutic
agents~\cite{KomarovaWodarz2014,WodarzKomarova2014,SchattlerLedzewicz2015}.


\section*{References}
%

\appendix

\section*{Appendix}

\setcounter{equation}{0}
\renewcommand{\theequation}{A.\arabic{equation}}

\subsection*{Proof of Theorem~{\rm \ref{Thm_2}}}

Consider a solution of (\ref{Eq_10}),~(\ref{Eq_11}) defined on a time interval~$ [0, t_1) $ with
$ \: t_1 \, \in \, [0, +\infty) \, \cup \, \{ +\infty \} $. For all $ t \in [0, t_1) $ and $ i = \overline{1, n} $,
this solution satisfies
$$
\dot{u}_i(t) \:\: = \:\: \varphi \left( \sum_{j = 1}^n u_j(t) \right) \, \sum_{j = 1}^n g_{ij} \, u_j(t) \:\, - \:\,
d_i \, u_i(t),
$$
\begin{equation}
\begin{aligned}
\frac{\mathrm{d}}{\mathrm{d} t} \left( u_i(t) \: \exp \left\{ -g_{ii} \, \int\limits_0^t
\varphi \left( \sum_{j = 1}^n u_j(\tau) \right) \, \mathrm{d} \tau \,\, + \,\, d_i t \right\} \right) & \\
= \:\: \exp \left\{ -g_{ii} \, \int\limits_0^t  \varphi \left( \sum_{j = 1}^n u_j(\tau) \right) \, \mathrm{d} \tau \,\, + \,\,
d_i t \right\} & \\
\cdot \:\, \varphi \left( \sum_{j = 1}^n u_j(t) \right) \, \sum_{\substack{j = 1, \\ j \neq i \:\,}}^n g_{ij} \, u_j(t). &
\end{aligned}  \label{Eq_A_1}
\end{equation}
From the relations~(\ref{Eq_8}), (\ref{Eq_11}), (\ref{Eq_A_1}) and nonnegativity of $ \varphi $, one concludes that
$$
\frac{\mathrm{d}}{\mathrm{d} t} \left( u_i(t) \: \exp \left\{ -g_{ii} \, \int\limits_0^t
\varphi \left( \sum_{j = 1}^n u_j(\tau) \right) \, \mathrm{d} \tau \,\, + \,\, d_i t \right\} \right) \:\: \geqslant \:\: 0
$$
and $ u_i(t) \geqslant 0 $ for all $ t \in [0, t_1) $ and $ i = \overline{1, n} $, which leads to the first statement of
the theorem.

In order to verify the second statement, it remains to note that the properties of $ \varphi $ given in
Assumption~\ref{Ass_1} imply the boundedness of the elements of the Jacobian matrix~(\ref{Eq_13}) on
$ [0, +\infty)^{n \times 1} $.  \qed

\subsection*{Proof of Theorem~{\rm \ref{Thm_5}}}

Consider a solution of (\ref{Eq_10}),~(\ref{Eq_11}) and the function~(\ref{Eq_15}). By virtue of Theorem~\ref{Thm_2}, one has
\begin{equation}
u_i(t) \, \geqslant \, 0, \quad i = \overline{1, n}, \quad s(t) \, \geqslant \, 0 \qquad \forall t \geqslant 0.
\label{Eq_A_2}
\end{equation}
Denote $ \: m_{\max} \: = \: \max_{i \, = \, \overline{1, n}} \, m_i $. Due to the relations~(\ref{Eq_7}), (\ref{Eq_10}),
(\ref{Eq_A_2}), and nonnegativity of $ \varphi $, one obtains
$$
\begin{aligned}
& \dot{s}(t) \:\, = \:\, \varphi(s(t)) \, \sum_{j = 1}^n m_j \, u_j(t) \,\, - \,\, \sum_{i = 1}^n d_i \, u_i(t) \\
& \qquad
\leqslant \:\, m_{\max} \, s(t) \, \varphi(s(t)) \,\, - \,\, d_{\min} \, s(t) \\
& \forall t \geqslant 0.
\end{aligned}
$$
Furthermore, Assumption~\ref{Ass_1} yields that
$$
0 \: \leqslant \: c \: = \: \sup_{\xi \, \in \, [0, +\infty)} \, \{ \xi \, \varphi(\xi) \} \: < \: +\infty.
$$
Hence,
\begin{equation}
\dot{s}(t) \,\, \leqslant \,\, c \, | m_{\max} | \: - \: d_{\min} \, s(t) \quad \forall t \geqslant 0.  \label{Eq_A_3}
\end{equation}
From the relations~(\ref{Eq_A_2}), (\ref{Eq_A_3}) and Assumption~\ref{Ass_4}, one concludes that $ s $ is bounded on
the whole time interval~$ [0, +\infty) $.  \qed

\subsection*{Proof of Theorem~{\rm \ref{Thm_27}}}

Let $ v \in V $ and $ h \in V $. Note that
\begin{equation}
\begin{aligned}
\varphi(\psi[v + h]) \:\, & = \:\, \varphi(\psi[v] \, + \, \psi[h]) \\
& = \:\, \varphi(\psi[v]) \,\, + \,\, \varphi'(\psi[v]) \, \psi[h] \,\, + \,\, o_1(h),
\end{aligned}  \label{Eq_A_4}
\end{equation}
where $ \, o_1 \colon V \to V \, $ is such that
$ \: \lim_{\| w \|_V \, \searrow \, 0} \, (\| o_1(w) \|_V \, / \, \| w \|_V) \: = \: 0 $. After substituting
(\ref{Eq_A_4}) into
$$
\begin{aligned}
\mathcal{F}[v + h] \:\: & = \:\: \varphi(\psi[v + h]) \: (\mathcal{G}[v + h] \: + \: \mathcal{B}[v + h]) \:\, - \:\,
\mathcal{D}[v + h] \\
& = \:\: \varphi(\psi[v + h]) \: (\mathcal{G}[v] \: + \: \mathcal{B}[v] \: + \: \mathcal{G}[h] \: + \: \mathcal{B}[h])
\:\, - \:\, \mathcal{D}[v] \:\, - \:\, \mathcal{D}[h],
\end{aligned}
$$
one obtains
$$
\begin{aligned}
\mathcal{F}[v + h] \:\: = \: & \,\, \mathcal{F}[v] \:\, + \:\, \varphi(\psi[v]) \: (\mathcal{G}[h] \: + \: \mathcal{B}[h]) \\
& + \:\, \varphi'(\psi[v]) \: \psi[h] \: (\mathcal{G}[v] \: + \: \mathcal{B}[v]) \:\, - \:\, \mathcal{D}[h] \:\, + \:\, o_2(h)
\end{aligned}
$$
with $ \, o_2 \colon V \to V \, $ satisfying
$ \: \lim_{\| w \|_V \, \searrow \, 0} \, (\| o_2(w) \|_V \, / \, \| w \|_V) \: = \: 0 $. For completing
the proof, it remains to use the definition of the Fr\'echet derivative.  \qed

\subsection*{Proof of Theorem~{\rm \ref{Thm_28}}}

In order to obtain the first statement of the theorem, note that the equation~(\ref{Eq_40}) can be transformed into
$$
\begin{aligned}
\frac{\partial}{\partial t} \left( u(x, t) \: \exp \left\{ -b(x) \int\limits_0^t
\varphi \left( \int\limits_{\Omega} u(y, \tau) \, \mathrm{d} y \right) \, \mathrm{d} \tau \:\, + \:\,
d(x) \, t \right\} \right) & \\
= \:\: \exp \left\{ -b(x) \int\limits_0^t \varphi \left( \int\limits_{\Omega} u(y, \tau) \, \mathrm{d} y \right) \,
\mathrm{d} \tau \:\, + \:\, d(x) \, t \right\} & \\
\cdot \:\, \varphi \left( \int\limits_{\Omega} u(y, t) \, \mathrm{d} y \right) \:
\int\limits_{\Omega} G(x, y) \, u(y, t) \, \mathrm{d} y &
\end{aligned}
$$
for all $ x \in \Omega $ and for all $ t $ on a time interval~$ [0, t_1) $ where a considered solution exists
(this transformation still makes sense when (\ref{Eq_40}) is understood as a differential equation with
the state space~$ V $), and that the functions~$ \varphi, G $ are nonnegative (due to Assumption~\ref{Ass_26}).

The second statement is established with the help of the following arguments:
\begin{list}{\rm \arabic{count})}%
{\usecounter{count}}
\item  $ V_+ $ is positively invariant with respect to the dynamical system~(\ref{Eq_40}) as verified above,
and it is not difficult to prove that $ V_+ $ is also a closed convex subset of $ \: V \, = \, L^1(\Omega; \mathbb{R}) $;
\item  the boundedness of the linear operators $ \mathcal{G}, \mathcal{B}, \mathcal{D} $ (ensured by
Assumptions~\ref{Ass_25}, \ref{Ass_26}) and the adopted properties of $ \varphi $ (see Assumption~\ref{Ass_1}
mentioned in Assumption~\ref{Ass_26}) imply the uniform boundedness of
$ \, \| \mathrm{D} \mathcal{F} [v] \|_{\mathcal{L}(V, V)} \, $ for $ v \in V_+ $;
\item  Item~2 and \cite[Remark~I.9.2]{DaleckiiKrein1974} yield that the right-hand side
operator~$ \mathcal{F} $ is Lipschitz continuous on $ V_+ $;
\item  due to Items~1 and 3, the reasonings in the proof of \cite[Theorem~VII.1.2]{DaleckiiKrein1974}
(relying on the contraction principle in \cite[Theorem~I.9.1]{DaleckiiKrein1974}) can be applied in order to obtain
the sought-after existence result.
\end{list}
It has to be emphasized that the choice of the state space $ \: V \, = \, L^1(\Omega; \mathbb{R}) \: $
plays a crucial role for deriving the property in Item~2 (from (\ref{Eq_43}), one can see that
$ \, \psi[v] \, = \, \| v \|_V \, $  for $ v \in V_+ $).  \qed

\subsection*{Proof of Theorem~{\rm \ref{Thm_34}}}

One in fact arrives at the problem of finding a real dominant eigenvalue of $ \mathcal{G}_1 + \mathcal{B}_1 $ with
a nonnegative eigenfunction when investigating a nontrivial steady state of the integro-differential equation
$$
\begin{aligned}
& \frac{\partial p(x, t)}{\partial t} \:\: = \:\: \mathcal{G}_1 [p(\cdot, t)] (x) \:\, + \:\,
\mathcal{B}_1 [p(\cdot, t)] (x) \\
& \qquad\qquad\quad \:\:
- \:\, p(x, t) \: \int\limits_{\Omega} (\mathcal{G}_1 [p(\cdot, t)] (y) \,\, + \,\,
\mathcal{B}_1 [p(\cdot, t)] (y)) \, \mathrm{d} y, \\
& x \in \Omega, \quad t \geqslant 0,
\end{aligned}
$$
with respect to a dynamical probability density~$ p(\cdot, t) $ satisfying $ \: \int_{\Omega} p(x, t) \, \mathrm{d} x \: = \: 1 \: $
for all $ t \geqslant 0 $. This equation can be rewritten as
$$
\begin{aligned}
& \frac{\partial p(x, t)}{\partial t} \:\: = \:\: \mathcal{G}_1 [p(\cdot, t)] (x) \:\, + \:\,
\tilde{\mathcal{B}}_1 [p(\cdot, t)] (x) \\
& \qquad\qquad\quad \:\:
- \:\, p(x, t) \: \int\limits_{\Omega} (\mathcal{G}_1 [p(\cdot, t)] (y) \,\, + \,\,
\tilde{\mathcal{B}}_1 [p(\cdot, t)] (y)) \, \mathrm{d} y, \\
& x \in \Omega, \quad t \geqslant 0,
\end{aligned}
$$
where
$$
\begin{aligned}
& \tilde{\mathcal{B}}_1 [p(\cdot, t)] (x) \:\: = \:\: \mathcal{B}_1 [p(\cdot, t)] (x) \:\, - \:\,
\left( \mathrm{ess} \, \sup_{y \, \in \, \Omega} \, \frac{b(y)}{d(y)} \right) \, p(x, t) \\
& \qquad\qquad\quad \:\:\:\,
= \:\: \left( \frac{b(x)}{d(x)} \,\, - \,\, \mathrm{ess} \, \sup_{y \, \in \, \Omega} \, \frac{b(y)}{d(y)} \right) \,
p(x, t) \\
& \forall x \in \Omega \quad \forall t \geqslant 0
\end{aligned}
$$
(recall that, according to (\ref{Eq_37}) and (\ref{Eq_38}), one has $ \, b(x) = 0 \, $ in the Eigen case and
$ \: b(x) \, = \, m(x) - \mu(x) \: $ in the Crow--Kimura case). Such a type of integro-differential equations was
studied in \cite{Burger1988a}. It then remains to use the results of \cite[\S 3]{Burger1988a} and to take into
account that \cite[Proposition~3.4]{Burger1988a} can be refined as was mentioned in the end of
\cite[\S IV.3]{Burger2000} (Items~3 and 4 of Assumption~\ref{Ass_31} play a significant role for applying that
refinement).  \qed

\end{document}